\documentclass[submission,copyright,creativecommons]{eptcs}
\providecommand{\event}{EXPRESS/SOS 2019} 
\usepackage{eptcs-preamble}

\usepackage{net,bussproofs,cmll}

\usepackage[toc,page]{appendix}
\usepackage{tikz,tikz-cats,tikz-cd}
\usetikzlibrary{matrix}

\renewcommand{\impll}{\rightarrow }

\renewcommand{\epsilon}{\varepsilon}
\renewcommand{\phi}{\varphi}

\newlist{enumerati}{enumerate}{10}
\setlist[enumerati]{label=\emph{(\roman*)}, ref=\emph{(\roman*)}}

\title{Cellular Monads from Positive GSOS Specifications}
\author{Tom Hirschowitz\thanks{Thanks to Jorge P\'erez and Jurriaan Rot for the invitation, and to the referees for helpful comments.}
  \institute{Univ. Grenoble Alpes, Univ. Savoie Mont Blanc, CNRS, LAMA\\ 73000 Chamb\'ery, France}
  \email{tom.hirschowitz@univ-smb.fr}
}
\def\titlerunning{Cellular Monads from Positive GSOS Specifications}
\def\authorrunning{T. Hirschowitz}
\begin{document}
\maketitle

\begin{abstract}
  We give a leisurely introduction to our abstract framework for
  operational semantics based on cellular monads on transition
  categories. Furthermore, we relate it for the first time to an
  existing format, by showing that all Positive GSOS specifications
  generate cellular monads whose free algebras are all
  compositional. As a consequence, we recover the known result that
  bisimilarity is a congruence in the generated labelled transition
  system.
\end{abstract}

\section{Introduction}\label{sec:introduction}
\subsection{Motivation}
In the vast majority of foundational research on programming
languages, although ideas are thought of as widely applicable, they
are presented on \emph{one}, simple example.  Typically, there is a
tension between simplicity of exposition, leading to the minimal
language making the idea relevant, and significance, leading to the
most expressive one.  Strikingly, the scope of the idea is often
mostly clear to the experts, but no attempt is made at stating it
precisely.  The reason for this is that the mathematical concepts
needed for even only making such statements are lacking. Indeed, one
needs to be able to say something like: ``for all programming
languages of such shape, the following holds''. But there simply is no
widely accepted mathematical notion of programming language.

Such a general notion should account for both
\begin{enumerati}
\item the interaction between syntax and dynamics, as involved in,
  e.g., structural operational semantics~\cite{PlotkinSOS}, or in the
  statement of results like type soundness, congruence of program
  equivalence, or compiler correctness, and
\item denotational semantics, in the sense of including not only
  operational, syntactic models but also others, typically ones in
  which program equivalence is coarser.
\end{enumerati}
Typically, standard formats~\cite{mousavi2007sos} elude denotational
semantics, and are exclusively syntactic.  To our knowledge, the only
such proposals meeting all these criteria are \emph{functorial
  operational semantics}, a.k.a.\ \emph{bialgebraic
  semantics}~\cite{plotkin:turi:bialgebraic}, and a few
variants~\cite{DBLP:journals/tcs/CorradiniHM02,DBLP:conf/lics/Staton08}.
This approach has been deeply developed, and shown to extend smoothly
to various settings, e.g., non-deterministic and probabilistic
languages.  However, two important extensions have proved more
difficult.
\begin{itemize}
\item The treatment of languages with variable binding is
  significantly more technical than the basic
  setting~\cite{DBLP:conf/lics/FioreT01,DBLP:conf/lics/FioreS06,DBLP:conf/lics/Staton08}.
\item More importantly, the bialgebraic study of higher-order
  languages like the $\lambda $-calculus or the higher-order $\pi $-calculus is
  only in its infancy~\cite{Peressotti}.
\end{itemize}
This leaves some room for exploring potential alternatives. 

  \subsection{Context}
  In recent work~\cite{hirschowitz:hal-01815328}, a new approach to
  abstract operational semantics was proposed, and its expressive
  power was demonstrated by proving for the first time an abstract
  soundness result for \emph{bisimulation up to context} in the
  presence of variable binding.  Bisimulation up to context is an
  efficient technique~\cite[Chapter~6]{SangioRutten} for proving program
  equivalences, which had previously been proved correct in the
  bialgebraic setting~\cite{DBLP:conf/csl/BonchiPPR14}, but only
  without binding.  

Its novelty mainly resides in the following two technical features.
\begin{description}
\item[Transition categories] First, while standard operational
  semantics is based on \emph{labelled transition systems}, this is
  both generalised and abstracted over in the framework.
  \begin{description}
  \item[Generalisation] Indeed, in the examples, instead of standard
    labelled transition systems, we use a slight generalisation
    similar in spirit to~\cite{DBLP:conf/ifipTCS/Fiore00}, essentially
    from relations to graphs, i.e., possibly with several transitions
    between two states.  This simple, harmless generalisation brings
    in a lot of useful structure, typically that of a
    \emph{topos}~\cite{MM}, which is unavailable at this level in
    bialgebraic operational semantics.
  \item[Abstraction] In full generality, the framework takes as a
    parameter a \emph{transition category}, a typical example of which
    is given by such generalised transition systems.  For any object
    of a given transition category, bisimulation may be defined by
    lifting, following an idea from~\cite{DBLP:conf/lics/JoyalNW93}.
  \end{description}
\item[Combinatorial category theory] A second technical innovation is
  the use of advanced combinatorial category theory.  To start with,
  \emph{familial monads}~\cite{DBLP:journals/mscs/CarboniJ95}, or
  rather their recent \emph{cellular}
  variant~\cite{garner:hal-01246365}, provide a notion of evaluation
  context for both programs and transitions, at the abstract
  level. Standard reasoning by induction on context thus becomes
  simple algebraic calculation.  A second, crucial notion is
  cofibrantly generated factorisation systems, a notion from homotopy
  theory~\cite{Hovey,riehl} which, together with cellularity, allows
  for a conceptually simple, yet relevant characterisation of
  well-behaved transition contexts.
\end{description}
\advance\textheight 13.6pt

Each instance of the framework is then constructed as follows.
\begin{description}
\item[Type of transition system] The first step is the choice of a
  type of transition system, which may involve different kinds of
  states (e.g., initial or final ones), the set of labels to be put on
  transitions, etc.  Technically, this amounts to fixing a transition
  category ${\mathcal C} $.  This also fixes the relevant notion of bisimulation,
  hence bisimilarity.
\item[Transition rules] The second step consists in defining the
  dynamics of the considered language, which is usually specified
  through a set of inference rules.  This comes in as a monad $T$ on
  ${\mathcal C} $, whose algebras are essentially the transition systems
  satisfying the given inference rules.  The standard, syntactic
  transition system is typically the free algebra $T(0)$.  This fixes
  the relevant notion of context closure. In this setting, congruence
  of bisimilarity $\sim _X$ on a $T$-algebra $X$ is the fact that
  $T(\sim _X) \rightarrow  X^2 $ factors through ${\sim _X} \rightarrow  X^2 $
  (see~\eqref{eq:cong:bisim} on page~\pageref{eq:cong:bisim}).
\end{description}
One of the main
results~\cite[Corollary~4.30]{hirschowitz:hal-01815328} is that if the
considered algebra is \emph{compositional}, in the sense that its
structure map $T(X) \rightarrow  X$ is a functional bisimulation, and if the
monad $T$ satisfies an additional condition, then bisimilarity is
indeed a congruence. The latter condition is called
\emph{${\mathbf T} _s $-familiality} in~\cite{hirschowitz:hal-01815328}, but we
will here call it \emph{cellularity}, because it is a specialisation
of cellularity in the sense of~\cite{garner:hal-01246365} to familial
functors.  As mentioned above, a second main
result~\cite[Corollary~5.15]{hirschowitz:hal-01815328} is that under a
different condition called \emph{${\mathbf T} _s ^\vee $-familiality}, bisimulation up
to context is sound.

\subsection{Contribution}
One of the main issues with cellular monads $T$ on transition
categories ${\mathcal C} $ is the lack of an efficient generation mechanism, i.e.,
a mathematical construction that produces pairs $({\mathcal C} ,T)$ from more
basic data.  In this paper, we initiate the search for such generating
constructions by showing that an existing simple format,
\emph{Positive GSOS}~\cite{GSOS}, always produces cellular monads
whose free algebras are compositional. As a consequence, we recover
(Theorem~\ref{thm}) the known result that bisimilarity is a congruence
in all free algebras.

As this is an invited contribution, we briefly introduce the approach
at an expository, rather concrete level. In particular, the only
considered transition category is the one of generalised labelled
transition systems in the sense alluded to above.  Finally, our proofs
are meant to be instructive rather than fully detailed.

\subsection{Plan}
In \partie \ref{sec:labelled}, we explain our generalisation of labelled
transition systems, and bisimulation by lifting.  In
\partie \ref{sec:positive}, we recall Positive GSOS specifications $\Sigma $ and
show how they generate monads $T_\Sigma $.  In \partie \ref{sec:models}, we argue
that algebras for the obtained monad $T_\Sigma $ are a good notion of model
for the considered Positive GSOS specification.  We then state
congruence of bisimilarity in categorical terms, and quickly reduce it
to two key properties: (i) compositionality of the considered algebra
and (ii) preservation of functional bisimulations by $T_\Sigma $.

We deal with~(i) in
\partie \ref{sec:compositionality}, where we show that when $T_\Sigma $ is obtained
from a Positive GSOS specification, all free algebras are
compositional. In \partie \ref{sec:preserving}, we then attack (ii), by
further reducing it to familiality and cellularity.  The remaining
sections develop these ideas.

In \partie \ref{sec:familiality}, we define familiality for functors (as
opposed to monads), and show that $T_\Sigma $ is a familial functor. In
\partie \ref{sec:factorisation}, we establish some factorisation properties
of familial functors which were announced and used
in~\partie \ref{sec:models} to reduce congruence of bisimilarity to
compositionality and preservation of bisimulation.  We then introduce
cellularity in~\partie \ref{sec:cellularity}, and show that $T_\Sigma $ is indeed
cellular. Finally, we wrap up in~\partie \ref{sec:monads:familiality} by
defining familiality for monads (which is slightly more demanding than
for mere functors), and showing that $T_\Sigma $ does form a familial monad.
This fills a hole left open in~\partie \ref{sec:compositionality}, thus
allowing us to state the main theorem.

Finally, we conclude and give some perspective in~\partie \ref{sec:conclu}.

\subsection{Prerequisites}
We assume familiarity with basic category
theory~\cite{MacLane:cwm,LeinsterCats}, including categories,
functors, natural transformations, monads and their algebras, and the
Yoneda lemma.

\section{Labelled transition systems as presheaves}\label{sec:labelled}
\subsection{Generalised transition systems}
A standard SOS specification is given by a signature, plus a family of
transition rules over a fixed set ${\mathbb A} $ of labels.  The set ${\mathbb A} $ fixes
the relevant kind of transition system, and we interpret this by
constructing a corresponding category of (generalised) transition
systems.  Given any set ${\mathbb A} $, let $\Gamma _{\mathbb A} $ denote the graph with
  \begin{itemize}
  \item vertex set ${\mathbb A}  + 1$, i.e., vertices are elements of ${\mathbb A} $,
    denoted by $[a]$ for $a \in  {\mathbb A} $, plus a special vertex $\star $,
  \item two edges $s^a ,t^a : \star  \rightarrow  [a]$, for all $a \in  {\mathbb A} $.
  \end{itemize}
  Pictorially, $\Gamma _{\mathbb A} $ looks like this:
  \begin{center}
    \diag{%
      \dots  \& {[a]} \& \dots  \& \mbox{($a \in  {\mathbb A} $)} \\
      \& \star \rlap{.} %
    }{%
      (m-2-2) edge[labell={s^a },bend left] (m-1-2) %
      (m-2-2) edge[labelr={t^a },bend right] (m-1-2) %
    }
  \end{center}
  There are no composable edges in $\Gamma _A$, so, adding formal identity
  arrows, it readily forms a category, which we also denote by $\Gamma _A$.

\begin{definition}
  The \emph{category of transition systems induced by ${\mathbb A} $} is
  $\psh{\Gamma _{\mathbb A} }$, the category of presheaves over $\Gamma _{\mathbb A} $.
\end{definition}
To see what presheaves over $\Gamma _{\mathbb A} $ have to do with transition systems,
let us observe that a presheaf $X \in  \psh{\Gamma _{\mathbb A} }$ consists of a set
$X(\star )$ of \emph{states}, together with, for each $a \in  {\mathbb A} $, a set of
\emph{transitions} $e \in  X[a]$ with source and target maps
$X(s^a ),X(t^a ): X[a] \rightarrow  X(\star )$.  Our notion is thus only slightly more
general than standard labelled transition systems over ${\mathbb A} $, in that it
allows several transitions with the same label between two given
states.
\begin{remark}
  The category $\psh{\Gamma _{\mathbb A} }$ may be viewed as a category of labelled
  graphs.  Indeed, letting $\Omega _{\mathbb A} $ denote the one-vertex graph with ${\mathbb A} $
  loops on it, we have by well-known abstract nonsense an equivalence
  ${\mathbf G} {\mathbf p} {\mathbf h} /\Omega _{\mathbb A}  \simeq  \psh{\Gamma _{\mathbb A} }$ of categories.  (This is due to the fact that
  $\Gamma _{\mathbb A} $ is isomorphic to the \emph{category of elements} of $\Omega _{\mathbb A} $, see
  Definition~\ref{def:elts} below.)
\end{remark}
\begin{notation}
  For any $X \in  \psh{\Gamma _{\mathbb A} }$, we denote the action of morphisms in $\Gamma _{\mathbb A} $
  with a dot.  E.g., if $e \in  X[a]$, then $e \cdot  s^a  \in  X(\star )$ is its
  source. We also sometimes abbreviate $s^a $ and $t^a $ to just $s$ and
  $t$.
\end{notation}
\begin{example}
  For languages like CCS~\cite{Milner80}, we let ${\mathbb A}  = {\mathcal N}  + {\mathcal N}  + 1$
  denote the disjoint union of a fixed set ${\mathcal N} $ of \emph{channel names}
  with itself and the singleton $1$.  Elements of the first term are
  denoted by $\abar$, for $a \in  {\mathcal N} $, and are used for output
  transitions, while elements of the second term, simply denoted by
  $a$, are used for input transitions. Finally, the unique element of
  the third term is denoted by $\tau $ and used for silent transitions.
  E.g., the labelled transition system
  \begin{center}
  \diag{|(x)| x \& |(y)| y \&
    |(z)| z %
  }{%
    (y) edge[labela={\abar}] (x) %
    edge[bend left=10,labela={b}] (z) %
    edge[bend right=10,labelb={b}] (z) %
    (z) edge[loop right,labela={a}] (z)
  }
\end{center}
is modelled by the presheaf
  $X$ with
  \begin{center}
  $\begin{array}[t]{c} X(\star ) = \ens{x,y,z}
    \end{array}$
    \hfil
    $\begin{array}[t]{c}
       X(\abar) = \ens{e} \\
       X(b) = \ens{f,f'} \\
       X(a) = \ens{g} 
    \end{array}$
    \hfil
    $\begin{array}[t]{c}
       x = e \cdot  t \\
       y = e \cdot  s = f \cdot  s = f' \cdot  s \\
       z = f \cdot  t = f' \cdot  t = g \cdot  s = g \cdot  t.
    \end{array}$
  \end{center}
\end{example}

\subsection{Bisimulation}
Returning to generalised transition systems, we may define
bisimulation categorically in the following way.  Morphisms
$f: X \rightarrow  Y$, i.e., natural transformations, are the analogue in our
setting of standard functional simulations. Indeed, given any
transition $e: x \xto{a} x'$ in $X$, then $f(x)$ sure has an $a$
transition to some state related to $x'$: this is simply $f(e)$!  The
next step is to define an analogue of functional bisimulation.  For
this, let us observe that the base category $\Gamma _{\mathbb A} $ embeds into the
presheaf category $\psh{\Gamma _{\mathbb A} }$ -- this is just the Yoneda embedding
${\mathbf y} : \Gamma _{\mathbb A}  \rightarrow  \psh{\Gamma _{\mathbb A} }$, directly specialised to our setting for
readability:
\begin{itemize}
\item the state object $\star $ embeds as the one-vertex graph ${\mathbf y} _\star $ with no transition;
\item any transition object $[a]$ embeds as the graph ${\mathbf y} _{[a]}$ with
  one $a$-transition between two distinct vertices;
\item the morphisms $s^a , t^a : \star  \rightarrow  [a]$ embed as the morphisms
  ${\mathbf y} _\star  \rightarrow  {\mathbf y} _{[a]}$ picking up the source and target, respectively, of
  the given transition.
\end{itemize}
\begin{notation}
  We often omit ${\mathbf y} $, treating it as an implicit coercion.
\end{notation}
\begin{definition}
  Let $f: X \rightarrow  Y$ be a \emph{functional bisimulation} whenever all
  commuting squares as the solid part below, admit a (potentially
  non-unique) \emph{lifting} $k$ as shown, i.e., a morphism making
  both triangles commute.
  \begin{center}
    \diag{%
      \star  \& X \\
      {[a]} \& Y %
    }{%
      (m-1-1) edge[labela={x}] (m-1-2) %
      edge[labell={s^a }] (m-2-1) %
      (m-2-1) edge[labelb={e}] (m-2-2) %
      edge[dashed,labelal={k}] (m-1-2) %
      (m-1-2) edge[labelr={f}] (m-2-2) %
    }
  \end{center}
\end{definition}

Let us explain why this matches the standard definition.  In any such
square, $x$ is essentially the same as just a state in $X$, while $e$
is just an $a$-transition in $Y$.  Furthermore, the composite
$\star  \xto{s^a } [a] \xto{e} Y$ picks the source of $e$, so commutation of
the square says that the source $e \cdot  s^a $ of $e$ is in fact $f(x)$
(a.k.a.\ $f \circ  x$). So we are in the situation described by the solid
part below.
\begin{center}
      \diag{%
      x \& f(x) \\
      x' \& y' %
    }{%
      (m-1-1) edge[mapsto,labela={f}] (m-1-2) %
      edge[dashed,labell={k}] (m-2-1) %
      (m-2-1) edge[dashed,mapsto,labelb={f}] (m-2-2) %
      (m-1-2) edge[labelr={e}] (m-2-2) %
    }
\end{center}
Finding a lifting $k$ then amounts to finding an antecedent to $e$
whose source is $x$, as desired.

We finally recover the analogue of standard bisimulation relations.
\begin{definition}
  A \emph{bisimulation relation} on $X$ is a subobject $R \hookrightarrow  X^2 $ ($=$
  isomorphism class of monomorphisms into $X^2 $) whose projections
  $R \rightarrow  X$ are both functional bisimulations.
\end{definition}
In this case, the above diagram specialises to 
\begin{center}
  \diag{%
    (x_1 ,x_2 ) \& x_i  \\
    (x'_1 ,x'_2 ) \& x'_i \rlap{,} %
  }{%
    (m-1-1) edge[mapsto,labela={\pi _i }] (m-1-2) %
    edge[dashed,labell={(e_1 ,e_2 )}] (m-2-1) %
    (m-2-1) edge[dashed,mapsto,labelb={\pi _i }] (m-2-2) %
    (m-1-2) edge[labelr={e_i }] (m-2-2) %
  }
\end{center}
where $(x_1 ,x_2 ),(x'_1 ,x'_2 ) \in  R(\star )$, and $(e_1 ,e_2 ) \in  R[a]$.

Now, $\psh{\Gamma _{\mathbb A} }$, as a presheaf category, is very well-behaved, namely
it is a Grothendieck topos~\cite{MM}.  In particular, subobjects of
$X^2 $ form a (small) complete lattice, in which the union of a family
$R_i  \hookrightarrow  X^2 $ is computed by first taking the copairing $\sum _i  R_i  \rightarrow  X^2 $,
which is generally not monic, and then taking its image.  Furthermore,
bisimulation relations are closed under unions and so admit a maximum
element,
\emph{bisimilarity}~\cite[Proposition~3.14]{hirschowitz:hal-01815328}.

The presheaf category $\psh{\Gamma _{\mathbb A} }$ is thus only a slight generalisation
of standard labelled transition systems over ${\mathbb A} $, in which we have an
analogue of bisimulation, conveniently defined by lifting, and
bisimilarity.  Let us now consider the case where states are terms in
a certain language, and transitions are defined inductively by a set
of transition rules, i.e., operational semantics.

\section{Positive GSOS specifications as monads}\label{sec:positive}
Let us briefly recall the Positive GSOS format.  We fix a set ${\mathbb A} $ of
labels, and start from a \emph{signature} $\Sigma _0  = (O_0 ,E_0 )$ on ${\mathbf S} {\mathbf e} {\mathbf t} $,
i.e., a set $O_0 $ equipped with a map $E_0 : O_0  \rightarrow  {\mathbb N} $.

\begin{definition}
  A \emph{Positive GSOS rule} over $\Sigma _0 $ consists of
  \begin{itemize}
  \item an operation $f \in  O_0 $, say of arity $n = E_0 (f)$,
  \item a label $a \in  {\mathbb A} $,
  \item $n$ natural numbers $m_1 ,\dots ,m_n $, 
  \item for all $i \in  n$, $m_i $ labels $a_{i,1},\dots ,a_{i,m_i }$, and
  \item a term $t$ with $n + \sum _{i = 1}^n  m_i $ free variables.
  \end{itemize}
\end{definition}
In more standard form, such a rule is just
\begin{mathpar}
  \inferrule{ \dots  \\ x_i  \xto{a_{i,j}} y_{i,j} \\ \dots  \\ (i \in  n, j \in  m_i ) %
  }{ %
    f(x_1 ,\dots ,x_n ) \xto{a} t %
  }
\end{mathpar}
where the $x_i $'s and $y_{i,j}$'s are all distinct and denote the
potential free variables of $t$.

\begin{definition}
  A \emph{Positive GSOS specification} is a signature $\Sigma _0 $, together
  with a set $\Sigma _1 $ of Positive GSOS rules.
\end{definition}

Let us now describe how any Positive GSOS specification $\Sigma $ induces a monad
$T_\Sigma $ on $\psh{\Gamma _{\mathbb A} }$, starting with the action of $T_\Sigma $ on objects.
Given any $X \in  \psh{\Gamma _{\mathbb A} }$, the set $T_\Sigma (X)(\star )$ of states consists of
all $\Sigma _0 $-terms with variables in $X(\star )$, as defined by the
grammar
$$M,N \Coloneqq \llparenthesis u\rrparenthesis  \aalt f(M_1 ,\dots ,M_n ),$$
where $u$ ranges over $X(\star )$.  Similarly, each $T_\Sigma (X)[a]$ consists
of all transition proofs following the rules in $\Sigma _1 $, with axioms in
all $X[a']$'s. Formally, such proofs are constructed inductively from the following rules,
\begin{mathpar}
  \inferrule{ }{\llparenthesis e\rrparenthesis \mathrel{:} _X \llparenthesis e \cdot  s\rrparenthesis  \xto{a} \llparenthesis e \cdot  t\rrparenthesis }~(e \in  X[a])
  \and
  \inferrule{ \dots  \\ R_{i,j}\mathrel{:} _X M_i  \xto{a_{i,j}} M_{i,j} \\ \dots  \\ (i \in  n, j \in  m_i ) %
  }{ %
    \rho (R_{i,j})_{i \in  n, j \in  m_i } \mathrel{:} _X
    f(M_1 ,\dots ,M_n ) \xto{a} t[ (x_i  \mapsto  M_i , (y_{i,j} \mapsto  M_{i,j})_{j \in  m_i })_{i \in  n} ] %
  }
\end{mathpar}
where in the second rule $f \in O_0$ , $E_0 (f) = n$,
$\rho = (f,a,(m_i ,(a_{i,j})_{j \in m_i })_{i \in n},t) \in \Sigma
_1$.
When $m_i  = 0$, we want to keep track of $M_i $ in the transition proof,
so by convention the family $(R_{i,j})_{j \in  m_i }$ denotes just $M_i $.
In the sequel we simply call \emph{transitions} such transition proofs.
\begin{example}\label{ex:ccs}
  Let us consider the following simple CCS transition of depth
  $>1$, in any $T_{CCS}(X)[\tau ]$.
      \begin{mathpar}
      \inferrule*{
        \inferrule*{
          \inferrule*{
          }{
            \llparenthesis e_1 \rrparenthesis :_X \llparenthesis x_1 \rrparenthesis  \xto{\abar} \llparenthesis y_1 \rrparenthesis 
          }
        }{
          \mathit{lpar}(\llparenthesis e_1 \rrparenthesis ,\llparenthesis x_2 \rrparenthesis ):_X \llparenthesis x_1 \rrparenthesis |\llparenthesis x_2 \rrparenthesis  \xto{\abar} \llparenthesis y_1 \rrparenthesis |\llparenthesis x_2 \rrparenthesis 
        }
        \\
        \inferrule*{
        }{
          \llparenthesis e_2 \rrparenthesis :_X \llparenthesis x_3 \rrparenthesis  \xto{a} \llparenthesis y_2 \rrparenthesis 
        }
      }{
        \mathit{sync}(\mathit{lpar}(\llparenthesis e_1 \rrparenthesis ,\llparenthesis x_2 \rrparenthesis ),\llparenthesis e_2 \rrparenthesis ):_X (\llparenthesis x_1 \rrparenthesis |\llparenthesis x_2 \rrparenthesis )|\llparenthesis x_3 \rrparenthesis  \xto{\tau } (\llparenthesis y_1 \rrparenthesis |\llparenthesis x_2 \rrparenthesis )|\llparenthesis y_2 \rrparenthesis 
      }~,
    \end{mathpar}
    where $\mathit{lpar}$ and $\mathit{sync}$ denote the left parallel and
    synchronisation rules, $(e_1 : x_1  \xto{\abar} y_1 ) \in  X[\abar]$,
    $x_2  \in  X(\star )$, and $(e_2 : x_3  \xto{a} y_2 ) \in  X[a]$. 
\end{example}

The source and target of a transition $R \mathrel{:} _X M \xto{a} N$ are
$M$ and $N$, respectively, which ends the definition of $T_\Sigma $ on
objects.  On morphisms $f: X \rightarrow  Y$, $T_\Sigma (f)$ merely amounts to
renaming variables $\llparenthesis x\rrparenthesis $ and $\llparenthesis e\rrparenthesis $ to $\llparenthesis f(x)\rrparenthesis $ and $\llparenthesis f(e)\rrparenthesis $,
respectively. It thus remains to show that $T_\Sigma $ has monad
structure. The unit $\eta _X: X \rightarrow  T_\Sigma (X)$ is obviously given by $\llparenthesis -\rrparenthesis $,
while multiplication $\mu _X: T_\Sigma  (T_\Sigma  (X)) \rightarrow  T_\Sigma (X)$ is given
inductively by removing the outer layer of $\llparenthesis -\rrparenthesis $'s
$$\begin{array}[t]{lrcl}
    \mbox{on states} & \mu _X\llparenthesis M\rrparenthesis  & = & M \\
                     & \mu _X(f(M_1 ,\dots ,M_n )) & = & f(\mu _X(M_1 ),\dots ,\mu _X(M_n )) \\
    \mbox{and on transitions} & \mu _X\llparenthesis R\rrparenthesis  & = & R \\
                     & \mu _X(\rho (R_{i,j})_{i \in  n,i \in  m_i }) & = & \rho (\mu _X(R_{i,j}))_{i \in  n,i \in  m_i }.
  \end{array}$$

  \begin{lemma}\label{lem:monad}
    The natural transformations $\eta $ and $\mu $ equip $T_\Sigma $ with monad
    structure.
  \end{lemma}
  \begin{proof}
    A straightforward induction.
  \end{proof}

  \section{Models as algebras and congruence of bisimilarity}\label{sec:models}
  Algebras for $T_\Sigma $ readily give the right notion of model for
  the transition rules:
  \begin{definition}
    An \emph{algebra} for a monad $T$, or a \emph{$T$-algebra},
    consists of an object $X$, equipped with a morphism
    $\alpha : T(X) \rightarrow  X$ such that the following diagrams commute.
    \begin{center}
      \diag{%
        T (T (X)) \& T(X) \\
        T(X) \& X
      }{%
        (m-1-1) edge[labela={T(\alpha )}] (m-1-2) %
        edge[labell={\mu _X}] (m-2-1) %
        (m-2-1) edge[labelb={\alpha }] (m-2-2) %
        (m-1-2) edge[labelr={\alpha }] (m-2-2) %
      }
      \hfil 
      \diag{%
        X \& \& T(X) \\
        \& X
      }{%
        (m-1-1) edge[labela={\eta _X}] (m-1-3) %
        edge[identity] (m-2-2) %
        (m-1-3) edge[labelbr={\alpha }] (m-2-2)
      }
    \end{center}
  \end{definition}
  Thus, intuitively, a $T_\Sigma $-algebra is a transition system which
  is stable under the given operations and transition rules.

  We now would like to show that, under suitable hypotheses,
  bisimilarity for any given $T_\Sigma $-algebra $\alpha : T_\Sigma (X) \rightarrow  X$ is a
  congruence.  We may state this categorically by saying that the
  canonical morphism $T_\Sigma (\sim _X) \rightarrow  X^2 $ factors through $m: (\sim _X) \hookrightarrow  X^2 $, as
  in
  \begin{equation}
    \diag(1,2){%
      T_\Sigma (\sim _X) \& \& {\sim _X} \\
      T_\Sigma (X^2 ) \& (T_\Sigma  (X))^2  \& X^2 \rlap{.} %
    }{%
      (m-1-1) edge[dashed,labela={}] (m-1-3) %
      edge[labell={T_\Sigma (m_X)}] (m-2-1) %
      (m-2-1) edge[labelb={\langle T_\Sigma (\pi _1 ),T_\Sigma (\pi _2 )\rangle }] (m-2-2) %
      (m-2-2) edge[labelb={\alpha ^2 }] (m-2-3) %
      (m-1-3) edge[labelr={m}] (m-2-3) %
    }
    \label{eq:cong:bisim}
  \end{equation}
  Indeed, an element of $T_\Sigma (\sim _X)$ is a term $M$ whose free variables
  are pairs of bisimilar elements of $X$, which we write as
  $M((x_1 ,y_1 ),\dots ,(x_n ,y_n ))$, with $x_i  \sim _X y_i $ for all $i \in  n$. The
  morphism $\langle T_\Sigma (\pi _1 ),T_\Sigma (\pi _2 )\rangle $ maps this to the pair
  $$(M(x_1 ,\dots ,x_n ),M(y_1 ,\dots ,y_n )),$$
  which $\alpha ^2 $ then evaluates componentwise.  The given factorisation
  thus boils down to
  $$\alpha (M(x_1 ,\dots ,x_n )) \sim _X \alpha (M(y_1 ,\dots ,y_n ))$$
  for all $M$ and $x_1  \sim _X y_1 $,\dots  , $x_n  \sim _X y_n $, i.e., bisimilarity is a
  congruence.

  In order to prove such a property, it is sufficient to prove that
  $T_\Sigma $ preserves \emph{all} bisimulation relations, in the sense that
  if $m: R \hookrightarrow  X^2 $ is a bisimulation relation, then so is
  $$T_\Sigma (R) \xto{T_\Sigma (m)} T_\Sigma (X^2 ) \xto{\langle T_\Sigma (\pi _1 ),T_\Sigma (\pi _2 )\rangle } (T_\Sigma (X))^2  \xto{\alpha ^2 } X^2 $$
  (in the slightly generalised sense that its image is).
  Equivalently, an easy diagram chasing shows that it all boils down to
  $$T_\Sigma (R) \xto{T_\Sigma (m)} T_\Sigma (X^2 ) \xto{T_\Sigma (\pi _i )} T_\Sigma (X) \xto{\alpha } X$$
  being a functional bisimulation for $i \in  \ens{1,2}$.

  Finally, $\pi _i \circ m$ is a functional bisimulation by definition, and
  functional bisimulations are stable under composition, so it is
  sufficient to prove that
  \begin{enumerati}
  \item \label{item:compositionality} the considered algebra is
    \emph{compositional}, in the sense that its structure map
    $\alpha : T_\Sigma (X) \rightarrow  X$ is a functional bisimulation, and
  \item \label{item:preservation} $T_\Sigma $ preserves all functional
    bisimulations.
  \end{enumerati}
  Compositionality essentially means that transitions of any
  $\alpha (M(x_1 ,\dots ,x_n ))$ are all obtained by assembling transitions of the
  $x_i $'s. This is not always the case, even for free algebras:
  \begin{example}
    Consider a specification $\Sigma $ consisting of the unique rule
    \begin{mathpar}
      \inferrule{x \xto{a} y}{f(g(x)) \xto{a} f(g(y))}~,
    \end{mathpar}
    say $\rho $, where $f$ and $g$ are two unary operations.  Then the
    free algebra $\mu _1 : T_\Sigma (T_\Sigma (1)) \rightarrow  T_\Sigma (1)$ is not compositional.
    Indeed, $1$ contains a unique vertex, say $\star $, and a transition
    $b: \star  \xto{b} \star $ for all labels $b$. Thus, $T_\Sigma (1)$ contains a
    transition $\rho \llparenthesis a\rrparenthesis : f(g\llparenthesis \star \rrparenthesis ) \xto{a} f(g\llparenthesis \star \rrparenthesis )$.  But the term
    $f(g\llparenthesis \star \rrparenthesis )$ is the image under $\mu _1 $ of $f\llparenthesis g\llparenthesis \star \rrparenthesis \rrparenthesis $, which has no
    transition.
  \end{example}

  Summing up, we have proved:
  \begin{lemma}\label{lem:abstract}
    If $T_\Sigma $ preserves functional bisimulations, then bisimilarity in
    any compositional $T_\Sigma $-algebra is a congruence.
  \end{lemma}

  \section{Compositionality}\label{sec:compositionality}
  Let us first consider compositionality.  For a general algebra, we
  cannot do more than taking compositionality as a
  hypothesis. However, we can say something when the considered
  algebra is free:
  \begin{lemma}\label{lem:compositionality}
    The multiplication $\mu _X: T_\Sigma  (T_\Sigma (X)) \rightarrow  T_\Sigma (X)$ is a functional
    bisimulation.
  \end{lemma}
  \begin{proof}
    We will see below (Lemma~\ref{lem:TSigma:familial:monad}) that all
    naturality squares of $\mu $ are pullbacks. In particular, we have a
    pullback
    \begin{center}
      \Diag{%
        \pbk{m-2-1}{m-1-1}{m-1-2} %
      }{%
        T_\Sigma  (T_\Sigma  (X)) \&  T_\Sigma  (T_\Sigma  (1)) \\
        T_\Sigma  (X) \& T_\Sigma (1)\rlap{.} %
      }{%
        (m-1-1) edge[labela={T_\Sigma  (T_\Sigma (!))}] (m-1-2) %
        edge[labell={\mu _X}] (m-2-1) %
        (m-2-1) edge[labelb={T_\Sigma (!)}] (m-2-2) %
        (m-1-2) edge[labelr={\mu _1 }] (m-2-2) %
      }
    \end{center}
    But functional bisimulations are easily seen to be stable under
    pullback, so it is enough to show that $\mu _1 $ is a functional
    bisimulation.  We thus consider any term ${\mathbf M} $ whose free variables
    are in $T_\Sigma (1)(\star )$, i.e., are themselves terms over a single free
    variable, say $\star $, together with a transition
    $R: \mu _1 ({\mathbf M} ) \xto{a} N$. And we need to show that there exists a
    transition ${\mathbf R} : {\mathbf M}  \xto{a} {\mathbf N} $ whose free variables and axioms are in
    $T_\Sigma (1)$, such that $\mu _{[a]}({\mathbf R} ) = R$.  We proceed by induction on ${\mathbf M} $:
    \begin{itemize}
    \item If ${\mathbf M}  = \llparenthesis M\rrparenthesis $, then taking ${\mathbf R}  = \llparenthesis R\rrparenthesis $ does the job.
    \item Otherwise, ${\mathbf M}  = f({\mathbf M} _1 ,\dots ,{\mathbf M} _n )$, so $M = \mu _1 ({\mathbf M} ) = f(M_1 ,\dots ,M_n )$,
      with $M_i  = \mu _1 ({\mathbf M} _i )$ for all $i \in  n$.  But then, $R$ must have the
      form $\rho (R_{i,j})_{i \in  n, j \in  m_i }$, for a certain rule
      $\rho  = (f,a,(m_i ,(a_{i,j})_{j \in  m_j })_{i \in  n},t)$ of $\Sigma $.
      By induction hypothesis, we find
      for all $i \in  n$ and $j \in  m_i $ a transition
      $${\mathbf R} _{i,j}: {\mathbf M} _i  \xto{a_{i,j}} {\mathbf N} _{i,j},$$
      such that $\mu _1 ({\mathbf R} _{i,j}) = R_{i,j}$.  Thus,
      ${\mathbf R}  = \rho ({\mathbf R} _{i,j})_{i \in  n, j \in  m_i }$ does have ${\mathbf M} $ as its source,
      and furthermore satisfies $\mu _1 ({\mathbf R} ) = R$, as desired. \qedhere
    \end{itemize}
  \end{proof}

  \section{Preserving bisimulations through familiality and
    cellularity}\label{sec:preserving}
  Let us now consider~\ref{item:preservation}, i.e., the fact that
  $T_\Sigma $ preserves functional bisimulations. So we need to find a lifting
  to any commuting square of the form
  \begin{center}
    \diag{%
      \star  \& T_\Sigma (X) \\
      {[a]} \& T_\Sigma (Y)\rlap{,} %
    }{%
      (m-1-1) edge[labela={M}] (m-1-2) %
      edge[labell={s^a }] (m-2-1) %
      (m-2-1) edge[labelb={R}] (m-2-2) %
      (m-1-2) edge[labelr={T_\Sigma (f)}] (m-2-2) %
    }
  \end{center}
  for any functional bisimulation $f$.

  We will proceed in two steps: we will require $T_\Sigma $ to be first
  familial, and then cellular. Familiality will allow us to factor the
  given square as the solid part below left, while cellularity will
  ensure existence of a lifting $k$ as on the right.
  \begin{equation}
    \diag{%
      \star  \& T_\Sigma (A) \& T_\Sigma (X) \\
      {[a]} \& T_\Sigma (B) \& T_\Sigma (Y) %
    }{%
      (m-1-1) edge[labela={M'}] (m-1-2) %
      edge[labell={s^a }] (m-2-1) %
      edge[bend left,labela={M}] (m-1-3) %
      (m-2-1) edge[labelb={R'}] (m-2-2) %
      edge[bend right,labelb={R}] (m-2-3) %
      (m-1-2) edge[labell={T_\Sigma (\gamma )}] (m-2-2) %
      (m-1-2) edge[labela={T_\Sigma (\phi )}] (m-1-3) %
      (m-2-2) edge[labelb={T_\Sigma (\psi )}] (m-2-3) %
      edge[dashed,labelal={T_\Sigma (k)}] (m-1-3) %
      (m-1-3) edge[labelr={T_\Sigma (f)}] (m-2-3) %
    }
    \qquad \qquad
    \diag{%
      A \& X \\
      B \& Y %
    }{%
      (m-1-1) edge[labela={\phi }] (m-1-2) %
      edge[labell={\gamma }] (m-2-1) %
      (m-2-1) edge[labelb={\psi }] (m-2-2) %
      edge[dashed,labelal={k}] (m-1-2) %
      (m-1-2) edge[labelr={f}] (m-2-2) %
    }
    \label{eq:congruence:proof:sketch}
  \end{equation}
  The composite $T_\Sigma (k) \circ  R'$ will thus give the desired lifting for the
  original square.

  At this stage, both steps may seem mysterious to the reader.  In
  fact, as we will see, factorisation as above left follows directly
  from the fact that $T_\Sigma $ may be expressed as a sum of representable
  functors. Let us first explain intuitively why this latter fact
  holds. We will then prove it more rigorously
  in~\partie \ref{sec:familiality}, to eventually return to factorisation
  in~\partie \ref{sec:factorisation}.

  To start with, let us observe that the set $T_\Sigma (1)(\star )$ consists of
  terms over a single free variable, say $\star $.  For any such term $M$,
  we may count the number of occurrences of $\star $, say $n_M$. Thus, any
  term in any $T_\Sigma (X)(\star )$ is entirely determined by an
  $M \in  T_\Sigma (1)(\star )$, together with a map $n_M \rightarrow  X(\star )$ assigning an
  element of $X(\star )$ to each occurrence of $\star $ in $M$.  But maps
  $n_M \rightarrow  X(\star )$ in ${\mathbf S} {\mathbf e} {\mathbf t} $ are in 1-1 correspondence with maps
  $n_M \cdot  {\mathbf y} _\star  \rightarrow  X$ in $\psh{\Gamma _{\mathbb A} }$, where $n_M \cdot  {\mathbf y} _\star $ denotes the
  $n_M$-fold coproduct ${\mathbf y} _\star  + \dots  + {\mathbf y} _\star $ of ${\mathbf y} _\star $ with itself.  In other
  words, letting $E^\star (M) = n_M \cdot  {\mathbf y} _\star $, we have
  \begin{equation}
    T_\Sigma (X)(\star ) \cong  \sum _{M \in  T_\Sigma (1)(\star )} \psh{\Gamma _{\mathbb A} }(E^\star (M),X).\label{eq:TSigma:fam:star}
  \end{equation}
  Clearly, for any $f: X \rightarrow  Y$, the action of $T_\Sigma (f)$ at $\star $
  is given by postcomposing with $f$, i.e., we have
  \begin{center}
    \diag(1,1.5){%
      T_\Sigma (X)(\star ) \& \sum _{M \in  T_\Sigma (1)(\star )} \psh{\Gamma _{\mathbb A} }(E^\star (M),X) \\
      T_\Sigma (Y)(\star ) \& \sum _{M \in  T_\Sigma (1)(\star )} \psh{\Gamma _{\mathbb A} }(E^\star (M),Y). %
    }{%
      (m-1-1) edge[iso,labela={}] (m-1-2) %
      edge[labell={T_\Sigma (f)_\star }] (m-2-1) %
      (m-2-1) edge[iso,labelb={}] (m-2-2) %
      (m-1-2) edge[shorten >=1ex,labelr={\sum _{M \in  T_\Sigma (1)(\star )} \psh{\Gamma _{\mathbb A} }(E^\star (M),f)}] (m-2-2) %
    }
  \end{center}
  The family~\eqref{eq:TSigma:fam:star} of isomorphisms is thus
  natural in $X$.  We will see shortly that this extends to objects
  other than $\star $.  Indeed, any transition in $T_\Sigma (X)[a]$ may be
  decomposed into a transition $R$ in $T_\Sigma (1)[a]$, together with a
  morphism $E^a (R) \rightarrow  X$, where $E^a (R)$ is obtained from occurrences of
  term and transition variables in $R$.

  We have seen that our isomorphisms are natural in $X$, so it seems
  natural to try to express some naturality constraint in the second
  argument of $T_\Sigma $.  But this requires making the right-hand side
  of~\eqref{eq:TSigma:fam:star} functorial in this variable in the
  first place!  In fact, for any transition $R: M \xto{a} N$, we will
  construct morphisms
  $$E^\star (M) \xto{E(s^a \restriction R)} E^a (R) \xot{E(t^a \restriction R)} E^\star (N)$$
  (see Notation~\ref{not:el:mor} below).  Thus, e.g., precomposing by
  the left-hand map yields the desired functorial action 
  $$\sum _{R \in  T_\Sigma (1)[a]} \psh{\Gamma _{\mathbb A} }(E^a (R),X) \rightarrow  \sum _{M \in  T_\Sigma (1)(\star )} \psh{\Gamma _{\mathbb A} }(E^\star (M),X),$$
  of $s^a : \star  \rightarrow  [a]$, sending any $\phi : E^a (R) \rightarrow  X$ to the composite
  \begin{equation}
    E^\star (M) \xto{E(s^a \restriction R)} E^a (R) \xto{\phi } X.\label{eq:action:s}
  \end{equation}

  \section{Familiality for functors}\label{sec:familiality}
  Let us now state more rigorously the definition of familiality and
  the fact that $T_\Sigma $ is familial. In the next section, we will
  explain how this entails the desired
  factorisation~\eqref{eq:congruence:proof:sketch}.
  \begin{definition}\label{def:elts}
    The \emph{category of elements} $el(X)$ of any presheaf $X \in  \psh{{\mathcal C} }$ on
    any category ${\mathcal C} $ has
    \begin{itemize}
    \item as objects all pairs $(c,x)$ with $c \in  {\mathbf o} {\mathbf b} ({\mathcal C} )$ and $x \in  X(c)$,
    \item and as morphisms $(c,x) \rightarrow  (c',x')$ all morphisms $f: c \rightarrow  c'$
      such that $x'\cdot f = x$.
    \end{itemize}
  \end{definition}
  \begin{notation}\label{not:el:mor}
    The morphism $f$, viewed as a morphism $(c,x) \rightarrow  (c',x')$, is entirely
    determined by $f$ and $x'$. We denote it by $f \restriction  x'$.
  \end{notation}

  \begin{definition}\label{def:familial:endofunctor}
    An endofunctor $F: \psh{{\mathbb C} } \rightarrow  \psh{{\mathbb C} }$ on a presheaf category is
    \emph{familial} iff there is a functor $E: el(F(1)) \rightarrow  \psh{{\mathbb C} }$
    such that
    \begin{equation}
      F(X)(c) \cong  \sum _{o \in  F(1)(c)} \psh{{\mathbb C} }(E(c,o),X),\label{eq:familiality:psh}
    \end{equation}
    naturally in $X \in  \psh{{\mathbb C} }$ and $c \in  {\mathbb C} $.
  \end{definition}

  And indeed, we have:
  \begin{lemma}\label{lem:TSigma:familial}
    The endofunctor $T_\Sigma $ is familial.
  \end{lemma}
  \begin{proof}
    We need to do two things: (1) extend the
    isomorphisms~\eqref{eq:TSigma:fam:star} to objects of the form
    $[a]$, and (2) define the morphisms $E(s^a \restriction R)$ and $E(t^a \restriction R)$
    rendering our isomorphisms natural also in the second argument of
    $T_\Sigma $.  In fact, we will do almost everything simultaneously by
    induction: we define $E^a (R)$ and $E(s^a \restriction R): E^\star (R \cdot  s^a ) \rightarrow  E^a (R)$ by
    induction on $R$.  By convention, as we did for the unique element
    $\star  \in  1(\star )$, we denote the unique element of $1[a]$ by $a$ itself.
    \begin{itemize}
    \item If $R = \llparenthesis a\rrparenthesis $, then its source is $M = \llparenthesis \star \rrparenthesis $ and we put
      $E^a (R) = {\mathbf y} _{[a]}$ and $E(s^a \restriction R) = s^a : {\mathbf y} _\star  \rightarrow  {\mathbf y} _{[a]}$.
    \item If $R = \rho (R_{i,j})_{i \in  n, j \in  m_i }:_1  M \xto{a} N$, then for all $i$ and
      $j \in  m_i $, by induction hypothesis, we get morphisms
      $$E(s^{a_{i,j}}\restriction R_{i,j}) : E^\star (M_i ) \rightarrow  E^{a_{i,j}}(R_{i,j}),$$
      where $R_{i,j}:_1  M_i  \xto{a_{i,j}} N_{i,j}$ for all $i \in  n$ and
      $j \in  m_i $.  Let us temporarily fix any $i \in  n$. For all
      $j,j' \in  m_i $, we have $R_{i,j} \cdot  s^a  = R_{i,j'} \cdot  s^a  = M_i $, so we
      take the wide pushout $E_i  = \bigoplus _{E^\star (M_i )} E^{a_{i,j}}(R_{i,j})$,
      i.e., the colimit of the following diagram.
      \begin{equation}
        \diag{%
          \& E^\star (M_i ) \\
          \dots  E^{a_{i,j}}(R_{i,j}) \& \dots  \& E^{a_{i,j'}}(R_{i,j'}) \dots 
        }{%
          (m-1-2) edge[labelal={E(s^{a_{i,j}}\restriction R_{i,j} )}] (m-2-1) %
          (m-1-2) edge[shorten >=1ex,labelar={E(s^{a_{i,j'}}\restriction R_{i,j'})}] (m-2-3) %
        }
        \label{eq:arity:rho}
      \end{equation}
      If $m_i  = 0$, this reduces to just $E^\star (M_i )$, which is exactly
      what we want.  Finally, we let $E^a (R)$ be the coproduct $\sum _i  E_i $
      of all the $E_i $'s, and observe that
      $E^\star (f(M_1 ,\dots ,M_n )) = \sum _i  E^\star (M_i )$ by definition, so that we may
      define $E(s^a \restriction R)$ to be the coproduct $\sum _i  S_i $ of all canonical
      injections $S_i : E^\star (M_i ) \rightarrow  E_i $.
    \end{itemize}

    This ends the inductive definition of $E^a (R)$ and $E(s^a \restriction R)$.  We
    now need to construct the morphisms $E(t^a \restriction R)$. We again proceed
    inductively.  When $R = \llparenthesis a\rrparenthesis $, the desired morphism is clearly $t^a $
    itself.  When $R = \rho (R_{i,j})_{i \in  n, j \in  m_i }$, the target is
    $N = t[x_i  \mapsto  M_i , (y_{i,j} \mapsto  N_{i,j})_{j \in  m_i }]$. Now, by construction, occurrences $occ_\star (N)$ of
    the unique variable $\star $ in $N$ are in 1-1 correspondence with
    $$V_N = \sum _i  \left ( (occ_\star (M_i ))^{occ_{x_i }(t)} + \sum _{j\in m_i } (occ_\star (N_{i,j}))^{occ_{y_{i,j}}(t)} \right ),$$
    and our map $E(t^a \restriction R)$ should reflect the intended correspondences.
    Since $E^\star (N) = V_N \cdot  {\mathbf y} _\star $,  $E(t^a \restriction R)$ is entirely determined
    by choosing a map $E_u : {\mathbf y} _\star  \rightarrow  E^a (R)$ for all $u \in  V_N$:
    \begin{itemize}
    \item If $u$ denotes an occurrence of $\star $ in $M_i $, for some
      occurrence of $x_i $ in $t$, we let $E_u $ denote the
      composite $${\mathbf y} _\star  \rightarrow  E^\star (M_i ) \rightarrow  E^\star (R),$$ where the latter map
      denotes injection into the colimit of~\eqref{eq:arity:rho}.
    \item If $u$ denotes an occurrence of $\star $ in $N_{i,j}$, for some
      occurrence of $y_{i,j}$ in $t$, we let $E_u $ denote the composite
      $${\mathbf y} _\star  \rightarrow  E^\star (N_{i,j}) \xto{E(t^{a_{i,j}}\restriction R_{i,j})}
      E^{a_{i,j}}(R_{i,j}) \rightarrow  E^\star (R),$$ where the latter map again
      denotes injection into the colimit of~\eqref{eq:arity:rho}. \qedhere
    \end{itemize}
  \end{proof}
  Rather than a full formal proof, let us illustrate that our
  construction satisfies the isomorphisms~\eqref{eq:TSigma:fam:star}
  on a few examples.
  \begin{example}
    In the case of the transition of Example~\ref{ex:ccs},
    familiality means that this transition is determined by picking the
    following transition in $T_{CCS}(1){[\tau ]}$,
    \begin{mathpar}
      \inferrule*{
        \inferrule*{
          \inferrule*{
          }{
            \llparenthesis \abar\rrparenthesis :_1  \llparenthesis \star \rrparenthesis  \xto{\abar} \llparenthesis \star \rrparenthesis 
          }
        }{
          \mathit{lpar}(\llparenthesis \abar\rrparenthesis ,\llparenthesis \star \rrparenthesis ):_1 \llparenthesis \star \rrparenthesis |\llparenthesis \star \rrparenthesis  \xto{\abar} \llparenthesis \star \rrparenthesis |\llparenthesis \star \rrparenthesis 
        }
        \\
        \inferrule*{
        }{
          \llparenthesis a\rrparenthesis :_1  \llparenthesis \star \rrparenthesis  \xto{a} \llparenthesis \star \rrparenthesis 
        }
      }{
        \mathit{sync}(\mathit{lpar}(\llparenthesis \abar\rrparenthesis ,\llparenthesis \star \rrparenthesis ),\llparenthesis a\rrparenthesis ):_1  (\llparenthesis \star \rrparenthesis |\llparenthesis \star \rrparenthesis )|\llparenthesis \star \rrparenthesis  \xto{\tau } (\llparenthesis \star \rrparenthesis |\llparenthesis \star \rrparenthesis )|\llparenthesis \star \rrparenthesis 
      }
    \end{mathpar}
    together with a morphism $E^\tau (\mathit{sync}(\mathit{lpar}(\llparenthesis \abar\rrparenthesis ,\llparenthesis \star \rrparenthesis ),\llparenthesis a\rrparenthesis )) \rightarrow  X$.
    Let us start with $E^\star (\mathit{lpar}(\llparenthesis \abar\rrparenthesis ,\llparenthesis \star \rrparenthesis ))$: it is given by the
    colimit of
    \begin{center}
      \diag{%
        {\mathbf y} _\star  \& {\mathbf y} _\star  \\
        {\mathbf y} _{[\abar]} %
      }{%
        (m-1-1) edge[labell={s^{\abar}}] (m-2-1) %
      }
    \end{center}
    (one ${\mathbf y} _\star $ for each argument $x_i $ of $\mathit{lpar}$, and for each $x_i $ one
    ${\mathbf y} _{[a_{i,j}]}$ for each premise $x_i  \xto{a_{i,j}} y_{i,j}$).
    Equivalently, this is just the coproduct ${\mathbf y} _{[\abar]} + {\mathbf y} _\star $, and
    $E(s^{\abar}\restriction \mathit{lpar}(\llparenthesis \abar\rrparenthesis ,\llparenthesis \star \rrparenthesis ))$ and
    $E(t^{\abar}\restriction \mathit{lpar}(\llparenthesis \abar\rrparenthesis ,\llparenthesis \star \rrparenthesis ))$ are given by
    $${{\mathbf y} _\star  + {\mathbf y} _\star } \xto{s^{\abar} + {\mathbf y} _\star } {{\mathbf y} _{[\abar]} + {\mathbf y} _\star } \xot{t^{\abar} + {\mathbf y} _\star } {{\mathbf y} _\star  + {\mathbf y} _\star .}$$
    It is then clear that $E(s^\tau \restriction \mathit{sync}(\mathit{lpar}(\llparenthesis \abar\rrparenthesis ,\llparenthesis \star \rrparenthesis ),\llparenthesis a\rrparenthesis ))$ is given by
    $${{\mathbf y} _\star  + {\mathbf y} _\star  + {\mathbf y} _\star } \xto{s^{\abar} + {\mathbf y} _\star  + s^a } {{\mathbf y} _{[\abar]} + {\mathbf y} _\star  + {\mathbf y} _{[a]}}.$$
    Now how about $E(t^\tau \restriction \mathit{sync}(\mathit{lpar}(\llparenthesis \abar\rrparenthesis ,\llparenthesis \star \rrparenthesis ),\llparenthesis a\rrparenthesis ))$? As the term $t$ occurring in the rule is here
    linear in the $y_{i,j}$'s, an easy computation leads to
    $${{\mathbf y} _{[\abar]} + {\mathbf y} _\star  + {\mathbf y} _{[a]}} \xot{t^{\abar} + {\mathbf y} _\star  + t^a } {{\mathbf y} _\star  + {\mathbf y} _\star  + {\mathbf y} _\star }.$$
    On this example, the isomorphism~\eqref{eq:familiality:psh} thus
    boils down to the transition
    $\mathit{sync}(\mathit{lpar}(\llparenthesis e_1 \rrparenthesis ,\llparenthesis x_2 \rrparenthesis ),\llparenthesis e_2 \rrparenthesis )$ above being entirely determined by
    picking $R = \mathit{sync}(\mathit{lpar}(\llparenthesis \abar\rrparenthesis ,\llparenthesis \star \rrparenthesis ),\llparenthesis a\rrparenthesis ) \in  T_{CCS}(1)[\tau ]$, and
    giving a morphism
    $$\phi : E^\tau (\mathit{sync}(\mathit{lpar}(\llparenthesis \abar\rrparenthesis ,\llparenthesis \star \rrparenthesis ),\llparenthesis a\rrparenthesis )) = {\mathbf y} _{[\abar]} + {\mathbf y} _\star  + {\mathbf y} _{[a]} \longrightarrow  X,$$
    which holds by universal property of coproduct and the Yoneda
    lemma.  Naturality of~\eqref{eq:familiality:psh} in $c$ says that
    the source of $(R,\phi )$ is given up to this correspondence by
    $(\llparenthesis \star \rrparenthesis |\llparenthesis \star \rrparenthesis )|\llparenthesis \star \rrparenthesis $ and the composite
    $$E^\star ((\llparenthesis \star \rrparenthesis |\llparenthesis \star \rrparenthesis )|\llparenthesis \star \rrparenthesis ) = {{\mathbf y} _\star  + {\mathbf y} _\star  + {\mathbf y} _\star } \xto{s^{\abar} + {\mathbf y} _\star  + s^a } {\mathbf y} _{[\abar]} + {\mathbf y} _\star  + {\mathbf y} _{[a]} \longrightarrow  X,$$    
    and likewise for the target.
  \end{example}

  \begin{example}
    Let us now illustrate the treatment of branching, in the sense of
    a rule having several premises involving the same $x_i $. An example
    from CCS is the `replicated synchronisation' rule
    \begin{mathpar}
      \inferrule{x_1  \xto{\abar} y_{1,1} \\ x_1  \xto{a} y_{1,2}}{{!}x_1  \xto{\tau } {!}x_1 |(y_{1,1}|y_{1,2})}~,
    \end{mathpar}
    say $\mathit{rsync}$.
    First, $E^\tau (\mathit{rsync}(\llparenthesis [\abar]\rrparenthesis ,\llparenthesis [a]\rrparenthesis ))$ is simply the pushout
    \begin{center}
      \Diag{%
        \pbk{m-2-1}{m-2-2}{m-1-2} %
      }{%
        {\mathbf y} _\star  \& {{\mathbf y} _{[a]}} \\
        {{\mathbf y} _{[\abar]}} \& E^\tau (\mathit{rsync}(\llparenthesis [\abar]\rrparenthesis ,\llparenthesis [a]\rrparenthesis ))\rlap{,} %
      }{%
        (m-1-1) edge[labela={s^a }] (m-1-2) %
        edge[labell={s^{\abar}}] (m-2-1) %
        (m-2-1) edge[labelb={}] (m-2-2) %
        (m-1-2) edge[labelr={}] (m-2-2) %
      }%
    \end{center}
    which rightly models the fact that a transition
    $\mathit{rsync}(\llparenthesis e_1 \rrparenthesis ,\llparenthesis e_2 \rrparenthesis ) \in  T_{CCS}(X)[\tau ]$ is entirely determined by
    picking $\mathit{rsync}(\llparenthesis [\abar]\rrparenthesis ,\llparenthesis [a]\rrparenthesis ) \in  T_{CCS}(1)[\tau ]$, together with
    elements $e_1 $ and $e_2 $ of $X[\abar]$ and $X[a]$ with a common
    source.
    The morphism $E(s^\tau \restriction \mathit{rsync}(\llparenthesis [\abar]\rrparenthesis ,\llparenthesis [a]\rrparenthesis ))$ is then straightforwardly given by
    the diagonal. The target morphism $E(t^\tau \restriction \mathit{rsync}(\llparenthesis [\abar]\rrparenthesis ,\llparenthesis [a]\rrparenthesis ))$ is a bit
    more complex to compute. Indeed, the target $t = {!}x_1 |(y_{1,1}|y_{1,2})$ has
    three free variables. The first, $x_1 $, should yield a morphism ${\mathbf y} _\star  \rightarrow  E^\tau (\mathit{rsync}(\llparenthesis [\abar]\rrparenthesis ,\llparenthesis [a]\rrparenthesis ))$
    that is determined by the source morphism $E^\star (M_1 ) \rightarrow  E^{\abar}(R_{1,1})$. Here,
    we get $${\mathbf y} _\star  \xto{s^{\abar}} E^{\abar}(\llparenthesis [\abar]\rrparenthesis ) = {\mathbf y} _{[\abar]} \rightarrow  E^\tau (\mathit{rsync}(\llparenthesis [\abar]\rrparenthesis ,\llparenthesis [a]\rrparenthesis )).$$
    On the other hand, $y_{1,1}$ and $y_{1,2}$ should be determined
    by the target morphisms $E^\star (M_{1,1}) \rightarrow  E^{\abar}(R_{1,1})$ and $E^\star (M_{1,2}) \rightarrow  E^a (R_{1,2})$,
    in our case
    \begin{center}
      ${\mathbf y} _\star  \xto{t^{\abar}} E^{\abar}(\llparenthesis [\abar]\rrparenthesis ) = {\mathbf y} _{[\abar]}  \rightarrow  E^\tau (\mathit{rsync}(\llparenthesis [\abar]\rrparenthesis ,\llparenthesis [a]\rrparenthesis ))$
      \hfil and \hfil
      ${\mathbf y} _\star  \xto{t^a } E^a (\llparenthesis a\rrparenthesis ) = {\mathbf y} _a   \rightarrow  E^\tau (\mathit{rsync}(\llparenthesis [\abar]\rrparenthesis ,\llparenthesis [a]\rrparenthesis ))$.
    \end{center}
  \end{example}

  \section{Familiality and factorisation}\label{sec:factorisation}
  Let us now return to our proof
  sketch~\eqref{eq:congruence:proof:sketch}, and explain the
  properties of familiality that allow us to factor the original
  square as indicated.  The crucial observation is that elements of
  the form
  $$(R,id_{E^c (R)}) \in  \sum _{R \in  T_\Sigma (1)(c)} \psh{\Gamma _{\mathbb A} }(E^c (R),E^c (R)) \cong  T_\Sigma (E^c (R))(c)$$
  have the special property that any other element of the form
  $(R,\phi ) \in  T_\Sigma (X)(c)$ may be obtained uniquely as the image of
  $(R,id)$ by the action of
  $$T_\Sigma (\phi )_c: T_\Sigma (E^c (R))(c) \rightarrow  T_\Sigma (X)(c).$$
  Having the same first component $R$ is equivalent to having the same
  image in $T_\Sigma (1)(c)$. So by Yoneda, having two elements of
  $T_\Sigma (X)(c)$ and $T_\Sigma (Y)(c)$ with common first component is the same
  as having a commuting square of the form below left.
  \begin{center}
    \diag{%
      {\mathbf y} _c \& T_\Sigma (X) \\
      T_\Sigma (Y) \& T_\Sigma (1) %
    }{%
      (m-1-1) edge[labela={}] (m-1-2) %
      edge[labell={}] (m-2-1) %
      (m-2-1) edge[labelb={}] (m-2-2) %
      (m-1-2) edge[labelr={}] (m-2-2) %
    }
    \hfil
    \diag{%
      {\mathbf y} _c \& T_\Sigma (X) \\
      T_\Sigma (E^c (R)) \& T_\Sigma (1) %
    }{%
      (m-1-1) edge[labela={p}] (m-1-2) %
      edge[labell={\xi }] (m-2-1) %
      (m-2-1) edge[labelb={T_\Sigma (!)}] (m-2-2) %
      edge[dashed,labelal={T_\Sigma (k)}] (m-1-2) %
      (m-1-2) edge[labelr={T_\Sigma (!)}] (m-2-2) %
    }
  \end{center}
  The special property of $(R,id)$ is thus equivalently that any
  commuting square as above right (solid part) admits a unique
  (dashed) lifting $k$ as shown, making the non-trivial triangle
  commute.  In fact, this holds more generally by replacing ${\mathbf y} _c$ and
  $1$ by arbitrary objects:
  \begin{definition}
    Given a functor $F: {\mathcal C}  \rightarrow  {\mathcal D} $,
    a morphism $\xi : D \rightarrow  F(C)$ is \emph{$F$-generic}, or \emph{generic} for short,
    when any commuting square as the solid part of 
    \begin{center}
      \diag{%
        D \& F(B) \\
        F(C) \& F(A) %
      }{%
        (m-1-1) edge[labela={\chi }] (m-1-2) %
        edge[labell={\xi }] (m-2-1) %
        (m-2-1) edge[labelb={F(k)}] (m-2-2) %
        edge[dashed,labelal={F(l)}] (m-1-2) %
        (m-1-2) edge[labelr={F(h)}] (m-2-2) %
      }
    \end{center}
    admits a unique \emph{strong} lifting $l$ as shown, in the
    sense that $F(l) \circ  \xi  = \chi $ and $h \circ  l = k$.
  \end{definition}

  \begin{lemma}[{\cite[Remark~2.12]{Weber:famfun}}]\label{lem:fam:gen}
    A functor $F: \psh{{\mathbb C} } \rightarrow  \psh{{\mathbb C} }$ is familial iff any morphism
    $Y \rightarrow  F(X)$ factors as
    $$Y \xto{\xi } F(A) \xto{F(\phi )} F(X),$$
    where $\xi $ is $F$-generic. This is called a \emph{generic-free} factorisation.
  \end{lemma}
  \begin{proof}[Proof sketch]
    $(\Rightarrow )$ Passing from ${\mathbf y} _c$ to any $Y$ goes by observing that generic
    morphisms are stable under colimits in the comma category
    $\psh{{\mathbb C} } \downarrow  F$, remembering that any presheaf $Z$ is a colimit of
    the composite $$el(Z) \xto{{\mathbf p} _Z} {\mathbb C}  \xto{{\mathbf y} } \psh{{\mathbb C} },$$ where ${\mathbf p} _Z$
    denotes the obvious projection functor.

    \noindent $(\Leftarrow )$ Conversely, $E(c,o)$ is given by $A$, for any choice of generic-free factorisation
    \begin{center}
      \hfill
      \diag|baseline=(m-2-1.base)|{%
        \& F(A) \\
        {\mathbf y} _c \& \& F(1)\rlap{.}
      }{%
        (m-2-1) edge[labelal={\xi }] (m-1-2) %
        edge[labelb={o}] (m-2-3) %
        (m-1-2) edge[labelar={F(\phi )}] (m-2-3) %
      }\qedhere
    \end{center}
  \end{proof}

  Lemma~\ref{lem:fam:gen} thus accounts for the factorisation of the
  original square as on the left in~\eqref{eq:congruence:proof:sketch}: $M$ and $R$
  respectively factor as
  \begin{center}
    $\star  \xto{M'} T_\Sigma (A) \xto{T_\Sigma (\phi )} T_\Sigma (X)$ \hfil and \hfil
    $\star  \xto{R'} T_\Sigma (B) \xto{T_\Sigma (\psi )} T_\Sigma (Y)$,
  \end{center}
  with $M'$ and $R'$ generic. But genericness of $M'$ yields the strong lifting $\gamma $ in
  \begin{center}
    \diag{%
      \star   \& {[a]} \& T_\Sigma (B) \\
      T_\Sigma (A) \& T_\Sigma (X) \& T_\Sigma (Y)\rlap{.} %
    }{%
      (m-1-1) edge[labela={s^a }] (m-1-2) %
      edge[labell={M'}] (m-2-1) %
      (m-2-1) edge[labelb={T_\Sigma (\phi )}] (m-2-2) %
      edge[dashed,labelal={T_\Sigma (\gamma )}] (m-1-3) %
      (m-1-2) edge[labela={R'}] (m-1-3) %
      (m-2-2) edge[labelb={T_\Sigma (f)}] (m-2-3) %
      (m-1-3) edge[labelr={T_\Sigma (\psi )}] (m-2-3) %
    }
  \end{center}

  \section{Cellularity}\label{sec:cellularity}
  We have now factored the original square as promised, but for the
  moment we have no guarantee that the `inner' square
  \begin{equation}
    \diag{%
      A \& X \\
      B \& Y %
    }{%
      (m-1-1) edge[labela={\phi }] (m-1-2) %
      edge[labell={\gamma }] (m-2-1) %
      (m-2-1) edge[labelb={\psi }] (m-2-2) %
      (m-1-2) edge[labelr={f}] (m-2-2) %
    }
    \label{eq:inner}
  \end{equation}
  will admit a lifting.  The point of cellularity is precisely this.
  For once, let us start from the abstract viewpoint and explain how
  directly relevant it is in this case.

  The starting point is the observation that our definition of
  bisimulation by lifting is based on a Galois connection.  Indeed,
  for any class ${\mathcal L} $ of morphisms, let ${\mathcal L} ^\boxslash $ denote the class of maps
  $f: X \rightarrow  Y$ such that for any $l: A \rightarrow  B$ in ${\mathcal L} $, any commuting square
  as below left admits a (not necessarily unique) lifting.
  \begin{center}
    \diag{%
      A \& X \\
      B \& Y %
    }{%
      (m-1-1) edge[labela={u}] (m-1-2) %
      edge[labell={{\mathcal L}  \ni  l}] (m-2-1) %
      (m-2-1) edge[labelb={v}] (m-2-2) %
      (m-1-2) edge[labelr={f \in  {\mathcal L} ^\boxslash }] (m-2-2) %
    }
    \hfil
    \diag{%
      X \& A \\
      Y \& B %
    }{%
      (m-1-1) edge[labela={u}] (m-1-2) %
      edge[labell={\wbotleft{{\mathcal R} } \ni  f}] (m-2-1) %
      (m-2-1) edge[labelb={v}] (m-2-2) %
      (m-1-2) edge[labelr={r \in  {\mathcal R} }] (m-2-2) %
    }
  \end{center}
  Conversely, given a class ${\mathcal R} $ of morphisms, let $\wbotleft{{\mathcal R} }$
  denote the class of morphisms $f: X \rightarrow  Y$ such that for any
  $r: A \rightarrow  B$ in ${\mathcal R} $, any commuting square as above right admits a
  lifting.  Clearly, letting ${\mathcal S} $ denote the set of all maps of the
  form $s^a : \star  \rightarrow  [a]$, ${\mathcal S} ^\boxslash $ catches exactly all functional
  bisimulations. But what is $\wbotrightleft{{\mathcal S} }$? In other words,
  which maps will admit a lifting against all functional
  bisimulations? This is very relevant to us, because finding a
  lifting for our inner square~\eqref{eq:inner} is obviously
  equivalent to showing that $\gamma  \in  \wbotrightleft{{\mathcal S} }$! Fortunately, the
  theory of weak factorisation systems gives a precise
  characterisation~\cite[Corollary~2.1.15]{Hovey}, of which we only
  need the following very special cases:
  \begin{lemma}\label{lem:wbotbot}
    Maps in $\wbotrightleft{{\mathcal S} }$ are closed under composition and
    pushout, in the sense that
    \begin{itemize}
    \item for any composable $f,g \in  \wbotrightleft{{\mathcal S} }$,
      $g \circ  f \in  \wbotrightleft{{\mathcal S} }$, and
    \item for any $f: X \rightarrow  Y$ in $\wbotrightleft{{\mathcal S} }$ and $u: X \rightarrow  X'$,
      the pushout $f'$ of $f$ along $u$, as below, is again in
      $\wbotrightleft{{\mathcal S} }$.
      \begin{center}
        \Diag(.6,2){%
          \pbk{m-2-1}{m-2-2}{m-1-2} %
        }{%
          X \& Y \\
          X' \& Y' %
        }{%
          (m-1-1) edge[labela={f \in  \wbotrightleft{{\mathcal S} }}] (m-1-2) %
          edge[labell={u}] (m-2-1) %
          (m-2-1) edge[labelb={f' \in  \wbotrightleft{{\mathcal S} }}] (m-2-2) %
          (m-1-2) edge[labelr={u'}] (m-2-2) %
        }
      \end{center}
    \end{itemize}
  \end{lemma}
  This is useful to us because the map $\gamma $ that we want to show is in
  $\wbotrightleft{{\mathcal S} }$ may be obtained as a finite composite of
  pushouts of maps in ${\mathcal S} $, which allows us to conclude.
  Indeed, $\gamma $ occurs in
  \begin{center}
    \diag{%
      \star  \& {[a]} \\
      T_\Sigma (E^\star (M)) \& T_\Sigma (E^a (R))\rlap{,} %
    }{%
      (m-1-1) edge[labela={s^a }] (m-1-2) %
      edge[labell={M'}] (m-2-1) %
      (m-2-1) edge[labelb={T_\Sigma (\gamma )}] (m-2-2) %
      (m-1-2) edge[labelr={R'}] (m-2-2) %
    }
  \end{center}
  with $M'$ and $R'$ generic.  So $(M',\gamma )$ is the generic-free
  factorisation of $R' \circ  s^a $ as in Lemma~\ref{lem:fam:gen}, hence,
  because generic-free factorisations are unique up to canonical
  isomorphism, we can actually compute $\gamma $. Indeed, letting $M' = (M'',id)$
  and $R' = (R'',id)$, for suitable $M'' \in  T_\Sigma (1)(\star )$ and
  $R'' \in  T_\Sigma (1)[a]$, by~\eqref{eq:action:s} $R' \circ  s^a $
  is the pair $(R'' \cdot  s^a , E(s^a \restriction R''))$, where
  $$E(s^a \restriction R''): E^\star (R'' \cdot  s^a ) \rightarrow  E^a (R'')$$ is obtained by familiality of
  $T_\Sigma $.  We thus get $$(M'',\gamma ) = (R''\cdot s^a , E(s^a \restriction R'')),$$ hence in
  particular $$\gamma  = E(s^a \restriction R'').$$ It is thus sufficient to show that
  each $E(s^a \restriction R)$ is in $\wbotrightleft{{\mathcal S} }$. This goes by induction on $R$,
  following an incremental construction of $E(s^a \restriction R)$. The base case is clear.
  When $R = \rho (R_{i,j})_{i \in  n, j \in  m_i }$, remember from the proof
  of Lemma~\ref{lem:TSigma:familial} that $E^a (R)$ is the coproduct
  $\sum _i  E_i $ for $i \in  n$, 
  each $E_i $ being constructed as the wide pushout of all
  $$E(s^{a_{i,j}}\restriction R_{i,j}): E^\star (M_i ) \rightarrow  E^{a_{i,j}}(R_{i,j}).$$
  Coproducts may be constructed by pushout along $0$, so it suffices
  to show that each diagonal $E^\star (M_i ) \rightarrow  E_i $ is in $\wbotrightleft{{\mathcal S} }$
  if each $E(s^{a_{i,j}}\restriction R_{i,j})$ is. This in turn boils down to
  incrementally constructing the diagonal $E^\star (M_i ) \rightarrow  E_i $ by
  successively pushing out each $E(s^{a_{i,j}}\restriction R_{i,j})$: assuming
  that we have constructed the diagonal $E^\star (M_i ) \rightarrow  E_i ^j $ up until
  $j < m_i $, we can incorporate $R_{i,j+1}$ by composing with the
  bottom morphism of
  \begin{center}
    \Diag(1,2){%
      \pbk[2em]{m-2-1}{m-2-2}{m-1-2}
    }{%
      E^\star (M_i ) \& E^{a_{i,j+1}}(R_{i,j+1}) \\
      E_i ^j  \& E_i ^{j+1}\rlap{,}
    }{%
      (m-1-1) edge[labela={E(s^{a_{i,j+1}}\restriction R_{i,j+1})}] (m-1-2) %
      edge[shorten >= 3pt,labell={}] (m-2-1) %
      (m-2-1) edge[labelb={}] (m-2-2) %
      (m-1-2) edge[shorten >= 5pt,labelr={}] (m-2-2) %
    }%
    \end{center}
    which is indeed in $\wbotrightleft{{\mathcal S} }$ by Lemma~\ref{lem:wbotbot}.
    Clearly, the obtained $E_i ^{m_i }$ is canonically isomorphic to $E_i $,
    so we have shown:
  \begin{lemma}\label{lem:cellular}
    The monad $T_\Sigma $ is \emph{cellular}, in the sense that in any
    commuting square of the form
    \begin{center}
      \diag{%
        \star  \& {[a]} \\
        T_\Sigma (A) \& T_\Sigma (B)\rlap{,} %
      }{%
        (m-1-1) edge[labela={s^a }] (m-1-2) %
        edge[labell={\xi }] (m-2-1) %
        (m-2-1) edge[labelb={T_\Sigma (\gamma )}] (m-2-2) %
        (m-1-2) edge[labelr={\chi }] (m-2-2) %
      }
    \end{center}
    with $\xi $ and $\chi $ generic, we have $\gamma  \in  \wbotrightleft{{\mathcal S} }$.
\end{lemma}

  \section{Familiality for monads}\label{sec:monads:familiality}
  We have now almost proved:
  \begin{lemma}\label{lem:preservation}
    $T_\Sigma $ preserves functional bisimulations.
  \end{lemma}
  The only remaining bit is the hole we left in the proof of
  Lemma~\ref{lem:compositionality}, when we claimed that all
  naturality squares of $\mu $ were pullbacks. Let us prove this now, as
  part of the following upgrade of
  Definition~\ref{def:familial:endofunctor} and
  Lemma~\ref{lem:TSigma:familial}.
  \begin{definition}
    A monad is \emph{familial} when its underlying functor is, and its
    unit and multiplication are \emph{cartesian} natural
    transformations, i.e., their naturality squares are pullbacks.
  \end{definition}
  In a case like ours, where the underlying category has a terminal
  object, by the pullback lemma, it is sufficient to verify that
  squares of the following form are pullbacks.
  \begin{center}
    \Diag{%
      \pbk{m-2-1}{m-1-1}{m-1-2} %
    }{%
      T_\Sigma ^2 (X) \& T_\Sigma ^2 (1) \\
      T_\Sigma (X) \& T_\Sigma (1)
    }{%
      (m-1-1) edge[labela={T_\Sigma ^2 (!)}] (m-1-2) %
      edge[labell={\mu _X}] (m-2-1) %
      (m-2-1) edge[labelb={T_\Sigma (!)}] (m-2-2) %
      (m-1-2) edge[labelr={\mu _1 }] (m-2-2) %
    }
    \hfil 
    \Diag{%
      \pbk{m-2-1}{m-1-1}{m-1-2} %
    }{%
      X \& 1 \\
      T_\Sigma (X) \& T_\Sigma (1)
    }{%
      (m-1-1) edge[labela={T_\Sigma ^2 (!)}] (m-1-2) %
      edge[labell={\eta _X}] (m-2-1) %
      (m-2-1) edge[labelb={T_\Sigma (!)}] (m-2-2) %
      (m-1-2) edge[labelr={\eta _1 }] (m-2-2) %
    }
  \end{center}
  The following will conclude our proof of congruence of bisimilarity:
  \begin{lemma}\label{lem:TSigma:familial:monad}
    $T_\Sigma $ is a familial monad.
  \end{lemma}
  \begin{proof}
    Pullbacks in presheaf categories being pointwise, we just need to
    check that a few types of squares are pullbacks in
    ${\mathbf S} {\mathbf e} {\mathbf t} $: for $\mu $ and $\eta $, and for each type of label.
    Let us treat the most interesting one, namely the left one below,
    assuming that the right one has already been covered.
    \begin{mathpar}
      \diag(.6,1){%
        T^2 (X)[a] \& T^2 (1)[a] \\
        T(X)[a] \& T(1)[a]
      }{%
        (m-1-1) edge[labela={T^2 (!)_{[a]}}] (m-1-2) %
        edge[labell={\mu _{X,{[a]}}}] (m-2-1) %
        (m-2-1) edge[labelb={T(!)_{[a]}}] (m-2-2) %
        (m-1-2) edge[labelr={\mu _{1,{[a]}}}] (m-2-2) %
      }
      \and 
      \diag(.6,1){%
        T^2 (X)(\star ) \& T^2 (1)(\star ) \\
        T(X)(\star ) \& T(1)(\star )
      }{%
        (m-1-1) edge[labela={T^2 (!)_\star }] (m-1-2) %
        edge[labell={\mu _{X,\star }}] (m-2-1) %
        (m-2-1) edge[labelb={T(!)_\star }] (m-2-2) %
        (m-1-2) edge[labelr={\mu _{1,\star }}] (m-2-2) %
      }
    \end{mathpar}
    Let $R[!]$ denote $T(!)(R)$ and ${\mathbf R} \llbracket !\rrbracket $ denote $T^2 (!)({\mathbf R} )$, for all
    $R \in  T(X)[a]$ and ${\mathbf R}  \in  T^2 (X)[a]$.  We must show that for any
    $a \in  {\mathbb A} $, given any ${\mathbf R}  \in  T^2 (1)[a]$ and $R \in  T(X)[a]$ such that
    $R[!] = \mu _1 ({\mathbf R} ),$ there exists a unique ${\mathbf R} ^0  \in  T^2 (X)[a]$ satisfying
    \begin{center}
      $\mu _X({\mathbf R} ^0 ) = R$ \hfil and \hfil ${\mathbf R} ^0 \llbracket !\rrbracket  = {\mathbf R} $.
    \end{center}
    
    We proceed by induction on ${\mathbf R} $. The base case is easy. For the induction step,
    if
      ${\mathbf R}  = \rho ({\mathbf R} _{i,j})_{i \in  n, j \in  m_i }$, then because
      $\mu _1 ({\mathbf R} ) = R[!]$, we have $R = \rho (R_{i,j})_{i \in  n, j \in  m_i }$ with
      $R_{i,j}[!] = \mu _1 ({\mathbf R} _{i,j})$ for all $i,j$\footnote{For all $i$ such
        that $m_i  = 0$, we in fact deal with some term $M_i $, using the
        corresponding square. Let us ignore this detail for readability.}.
      By induction hypothesis, we find
      a family ${\mathbf R} ^0 _{i,j} \in  T^2 (X)[a_{i,j}]$ such that
      \begin{center}
        $\mu _X({\mathbf R} ^0 _{i,j}) = R_{i,j}$ \hfil and \hfil ${\mathbf R} ^0 _{i,j}\llbracket !\rrbracket  = {\mathbf R} _{i,j}$
        \hfil for all $i,j$.
      \end{center}
      Letting now ${\mathbf R} ^0  = \rho ({\mathbf R} ^0 _{i,j})_{i \in  n, j \in  m_i }$, we get as
      desired
      \begin{center}
        $\mu _X({\mathbf R} ^0 ) = \rho (\mu _X({\mathbf R} ^0 _{i,j}))_{i,j} = \rho (R_{i,j})_{i,j} = R$ \hfil and \hfil
        ${\mathbf R} ^0 \llbracket !\rrbracket  = \rho ({\mathbf R} ^0 _{i,j}\llbracket !\rrbracket )_{i,j} = \rho ({\mathbf R} _{i,j})_{i,j} = {\mathbf R} $.
 \qedhere    
      \end{center}
  \end{proof}

  This ends the proof of:
  \begin{theorem}\label{thm}
    For all $X \in  \psh{\Gamma _{\mathbb A} }$ and Positive GSOS specifications $\Sigma $,
    bisimilarity in the free algebra $T_\Sigma (X)$ is a congruence.
  \end{theorem}
  
  \section{Conclusion and perspectives}\label{sec:conclu}
  In this paper, we have introduced the familial approach to
  programming language theory~\cite{hirschowitz:hal-01815328} at the
  rather concrete level of generalised labelled transition systems
  ($\psh{\Gamma _{\mathbb A} }$). Notably, we have recalled the notions of cellular
  monad and compositional algebra, and recalled that bisimilarity is
  always a congruence in a compositional algebra for a cellular monad
  (Lemma~\ref{lem:abstract}).

  We have also shown that all monads $T_\Sigma $ generated from a Positive
  GSOS specification $\Sigma $ are cellular (Lemma~\ref{lem:cellular}) and
  that free $T_\Sigma $-algebras are always compositional
  (Lemma~\ref{lem:compositionality}).  Putting all three results
  together, we readily recover (Theorem~\ref{thm}) the known result
  that bisimilarity is a congruence for all free $T_\Sigma $-algebras. In
  particular, this is the case for the standard, syntactic transition
  system, which is the initial algebra $T_\Sigma (0)$.

  This result constitutes a first generic tool for constructing
  instances of the framework of~\cite{hirschowitz:hal-01815328}.
  However, its scope is rather limited, and we plan to refine the
  construction to cover other formats like
  \emph{tyft/tyxt}~\cite{tyft}. A striking and promising observation
  here is that the well-foundedness condition demanded of a tyft/tyxt
  specification for bisimilarity to be a congruence is clearly covered
  by our approach based on weak factorisation systems
  (see~\partie \ref{sec:cellularity}). Cellularity thus provides a semantic
  criterion for well-foundedness, whose precise relationship with the
  original, syntactic one seems worth investigating.

  Beyond the task of showing by hand that existing formats yield
  cellular monads whose free algebras are compositional, we also plan
  to investigate a more categorical understanding of the generating
  process. The main motivation here is to design a construction that
  would cover variable binding. The theory developed by Fiore and his
  colleagues~\cite{DBLP:conf/lics/Fiore08,FioreHurEquational} seems
  like a good starting point.

  \bibliographystyle{./eptcs}
  \bibliography{./b}

\ifx\french\undefined \def\biling#1#2{#1} \else \def\biling#1#2{#2}
  \fi\ifx\french\undefined \ifx\spanish\undefined \def\triling#1#2#3{#1} \else
  \def\triling#1#2#3{#3} \fi\else \def\triling#1#2#3{#2}
  \fi\ifx\abbrevbib\undefined \def\abbrev#1#2{#1} \else \def\abbrev#1#2{#2} \fi
\begin{thebibliography}{10}
\providecommand{\bibitemdeclare}[2]{}
\providecommand{\surnamestart}{}
\providecommand{\surnameend}{}
\providecommand{\urlprefix}{Available at }
\providecommand{\url}[1]{\texttt{#1}}
\providecommand{\href}[2]{\texttt{#2}}
\providecommand{\urlalt}[2]{\href{#1}{#2}}
\providecommand{\doi}[1]{doi:\urlalt{http://dx.doi.org/#1}{#1}}
\providecommand{\bibinfo}[2]{#2}

\bibitemdeclare{article}{GSOS}
\bibitem{GSOS}
\bibinfo{author}{B.~\surnamestart Bloom\surnameend},
  \bibinfo{author}{S.~\surnamestart Istrail\surnameend} \&
  \bibinfo{author}{A.~\surnamestart Meyer\surnameend} (\bibinfo{year}{1995}):
  \emph{\bibinfo{title}{Bisimulation can't be traced}}.
\newblock {\sl \bibinfo{journal}{\abbrev{Journal of the ACM}{J. ACM}}}
  \bibinfo{volume}{42}, pp. \bibinfo{pages}{232--268},
  \doi{10.1145/200836.200876}.

\bibitemdeclare{inproceedings}{DBLP:conf/csl/BonchiPPR14}
\bibitem{DBLP:conf/csl/BonchiPPR14}
\bibinfo{author}{Filippo \surnamestart Bonchi\surnameend},
  \bibinfo{author}{Daniela \surnamestart Petrisan\surnameend},
  \bibinfo{author}{Damien \surnamestart Pous\surnameend} \&
  \bibinfo{author}{Jurriaan \surnamestart Rot\surnameend}
  (\bibinfo{year}{2014}): \emph{\bibinfo{title}{Coinduction up-to in a
  fibrational setting}}.
\newblock In: {\sl \bibinfo{booktitle}{Proc.\ 29th \abbrev{Symposium on Logic
  in Computer Science}{Logic in Comp. Sci.}}}, \bibinfo{publisher}{ACM}, pp.
  \bibinfo{pages}{20:1--20:9}, \doi{10.1145/2603088.2603149}.

\bibitemdeclare{article}{DBLP:journals/mscs/CarboniJ95}
\bibitem{DBLP:journals/mscs/CarboniJ95}
\bibinfo{author}{Aurelio \surnamestart Carboni\surnameend} \&
  \bibinfo{author}{Peter \surnamestart Johnstone\surnameend}
  (\bibinfo{year}{1995}): \emph{\bibinfo{title}{Connected Limits, Familial
  Representability and Artin Glueing}}.
\newblock {\sl \bibinfo{journal}{\abbrev{Mathematical Structures in Computer
  Science}{MSCS}}} \bibinfo{volume}{5}(\bibinfo{number}{4}), pp.
  \bibinfo{pages}{441--459}, \doi{10.1017/S0960129500001183}.

\bibitemdeclare{article}{DBLP:journals/tcs/CorradiniHM02}
\bibitem{DBLP:journals/tcs/CorradiniHM02}
\bibinfo{author}{Andrea \surnamestart Corradini\surnameend},
  \bibinfo{author}{Reiko \surnamestart Heckel\surnameend} \&
  \bibinfo{author}{Ugo \surnamestart Montanari\surnameend}
  (\bibinfo{year}{2002}): \emph{\bibinfo{title}{Compositional {SOS} and beyond:
  a coalgebraic view of open systems}}.
\newblock {\sl \bibinfo{journal}{\abbrev{Theoretical Computer Science}{Theor.
  Comp. Sci.}}} \bibinfo{volume}{280}(\bibinfo{number}{1-2}), pp.
  \bibinfo{pages}{163--192}, \doi{10.1016/S0304-3975(01)00025-1}.

\bibitemdeclare{article}{FioreHurEquational}
\bibitem{FioreHurEquational}
\bibinfo{author}{Marcelo \surnamestart Fiore\surnameend} \&
  \bibinfo{author}{Chung-Kil \surnamestart Hur\surnameend}
  (\bibinfo{year}{2009}): \emph{\bibinfo{title}{On the construction of free
  algebras for equational systems}}.
\newblock {\sl \bibinfo{journal}{\abbrev{Theoretical Computer Science}{Theor.
  Comp. Sci.}}} \bibinfo{volume}{410}, pp. \bibinfo{pages}{1704--1729},
  \doi{10.1016/j.tcs.2008.12.052}.

\bibitemdeclare{inproceedings}{DBLP:conf/ifipTCS/Fiore00}
\bibitem{DBLP:conf/ifipTCS/Fiore00}
\bibinfo{author}{Marcelo~P. \surnamestart Fiore\surnameend}
  (\bibinfo{year}{2000}): \emph{\bibinfo{title}{Fibred Models of Processes:
  Discrete, Continuous, and Hybrid Systems}}.
\newblock In: {\sl \bibinfo{booktitle}{IFIP TCS}}, {\sl \bibinfo{series}{LNCS}}
  \bibinfo{volume}{1872}, \bibinfo{publisher}{Springer}, pp.
  \bibinfo{pages}{457--473}, \doi{10.1007/3-540-44929-9\_32}.

\bibitemdeclare{inproceedings}{DBLP:conf/lics/Fiore08}
\bibitem{DBLP:conf/lics/Fiore08}
\bibinfo{author}{Marcelo~P. \surnamestart Fiore\surnameend}
  (\bibinfo{year}{2008}): \emph{\bibinfo{title}{Second-Order and
  Dependently-Sorted Abstract Syntax}}.
\newblock In: {\sl \bibinfo{booktitle}{LICS}},
  \bibinfo{organization}{\abbrev{IEEE}{IEEE}}, pp. \bibinfo{pages}{57--68},
  \doi{10.1109/LICS.2008.38}.

\bibitemdeclare{inproceedings}{DBLP:conf/lics/FioreS06}
\bibitem{DBLP:conf/lics/FioreS06}
\bibinfo{author}{Marcelo~P. \surnamestart Fiore\surnameend} \&
  \bibinfo{author}{Sam \surnamestart Staton\surnameend} (\bibinfo{year}{2006}):
  \emph{\bibinfo{title}{A Congruence Rule Format for Name-Passing Process
  Calculi from Mathematical Structural Operational Semantics}}.
\newblock In: {\sl \bibinfo{booktitle}{Proc.\ 21st \abbrev{Symposium on Logic
  in Computer Science}{Logic in Comp. Sci.}}},
  \bibinfo{organization}{\abbrev{IEEE}{IEEE}}, pp. \bibinfo{pages}{49--58},
  \doi{10.1109/LICS.2006.7}.

\bibitemdeclare{inproceedings}{DBLP:conf/lics/FioreT01}
\bibitem{DBLP:conf/lics/FioreT01}
\bibinfo{author}{Marcelo~P. \surnamestart Fiore\surnameend} \&
  \bibinfo{author}{Daniele \surnamestart Turi\surnameend}
  (\bibinfo{year}{2001}): \emph{\bibinfo{title}{Semantics of Name and Value
  Passing}}.
\newblock In: {\sl \bibinfo{booktitle}{Proc.\ 16th \abbrev{Symposium on Logic
  in Computer Science}{Logic in Comp. Sci.}}},
  \bibinfo{organization}{\abbrev{IEEE}{IEEE}}, pp. \bibinfo{pages}{93--104},
  \doi{10.1109/LICS.2001.932486}.

\bibitemdeclare{article}{garner:hal-01246365}
\bibitem{garner:hal-01246365}
\bibinfo{author}{Richard H.~G. \surnamestart Garner\surnameend} \&
  \bibinfo{author}{Tom \surnamestart Hirschowitz\surnameend}
  (\bibinfo{year}{2018}): \emph{\bibinfo{title}{Shapely monads and analytic
  functors}}.
\newblock {\sl \bibinfo{journal}{Journal of Logic and Computation}}
  \bibinfo{volume}{28}(\bibinfo{number}{1}), pp. \bibinfo{pages}{33--83},
  \doi{10.1093/logcom/exx029}.

\bibitemdeclare{article}{tyft}
\bibitem{tyft}
\bibinfo{author}{Jan~Friso \surnamestart Groote\surnameend} \&
  \bibinfo{author}{Frits \surnamestart Vaandrager\surnameend}
  (\bibinfo{year}{1992}): \emph{\bibinfo{title}{Structured Operational
  Semantics and Bisimulation as a Congruence}}.
\newblock {\sl \bibinfo{journal}{\abbrev{Information and Computation}{Inf. and
  Comp.}}} \bibinfo{volume}{100}, pp. \bibinfo{pages}{202--260},
  \doi{10.1016/0890-5401(92)90013-6}.

\bibitemdeclare{article}{hirschowitz:hal-01815328}
\bibitem{hirschowitz:hal-01815328}
\bibinfo{author}{Tom \surnamestart Hirschowitz\surnameend}
  (\bibinfo{year}{2019}): \emph{\bibinfo{title}{Familial monads and structural
  operational semantics}}.
\newblock {\sl \bibinfo{journal}{{PACMPL}}}
  \bibinfo{volume}{3}(\bibinfo{number}{{POPL}}), pp.
  \bibinfo{pages}{21:1--21:28}, \doi{10.1145/3290334}.

\bibitemdeclare{book}{Hovey}
\bibitem{Hovey}
\bibinfo{author}{Mark \surnamestart Hovey\surnameend} (\bibinfo{year}{1999}):
  \emph{\bibinfo{title}{Model Categories}}.
\newblock {\sl \bibinfo{series}{Mathematical Surveys and Monographs, Volume 63,
  AMS (1999)}}~\bibinfo{volume}{63}, \bibinfo{publisher}{American Mathematical
  Society}, \doi{10.1090/surv/063}.

\bibitemdeclare{inproceedings}{DBLP:conf/lics/JoyalNW93}
\bibitem{DBLP:conf/lics/JoyalNW93}
\bibinfo{author}{Andr{\'e} \surnamestart Joyal\surnameend},
  \bibinfo{author}{Mogens \surnamestart Nielsen\surnameend} \&
  \bibinfo{author}{Glynn \surnamestart Winskel\surnameend}
  (\bibinfo{year}{1993}): \emph{\bibinfo{title}{Bisimulation and open maps}}.
\newblock In: {\sl \bibinfo{booktitle}{Proc.\ 8th \abbrev{Symposium on Logic in
  Computer Science}{Logic in Comp. Sci.}}},
  \bibinfo{organization}{\abbrev{IEEE}{IEEE}}, pp. \bibinfo{pages}{418--427},
  \doi{10.1109/LICS.1993.287566}.

\bibitemdeclare{book}{LeinsterCats}
\bibitem{LeinsterCats}
\bibinfo{author}{Tom \surnamestart Leinster\surnameend} (\bibinfo{year}{2014}):
  \emph{\bibinfo{title}{Basic Category Theory}}.
\newblock {\sl \bibinfo{series}{Cambridge Studies in Advanced Mathematics}}
  \bibinfo{volume}{143}, \bibinfo{publisher}{Cambridge University Press},
  \doi{10.1017/CBO9781107360068}.

\bibitemdeclare{book}{MacLane:cwm}
\bibitem{MacLane:cwm}
\bibinfo{author}{Saunders \surnamestart {Mac Lane}\surnameend}
  (\bibinfo{year}{1998}): \emph{\bibinfo{title}{Categories for the Working
  Mathematician}}, \bibinfo{edition}{2nd} edition.
\newblock {\sl \bibinfo{series}{Graduate Texts in
  Mathematics}}~\bibinfo{volume}{5}, \bibinfo{publisher}{Springer},
  \doi{10.1007/978-1-4757-4721-8}.

\bibitemdeclare{book}{MM}
\bibitem{MM}
\bibinfo{author}{Saunders \surnamestart {Mac Lane}\surnameend} \&
  \bibinfo{author}{Ieke \surnamestart Moerdijk\surnameend}
  (\bibinfo{year}{1992}): \emph{\bibinfo{title}{Sheaves in Geometry and Logic:
  A First Introduction to Topos Theory}}.
\newblock \bibinfo{series}{Universitext}, \bibinfo{publisher}{Springer},
  \doi{10.1007/978-1-4612-0927-0}.

\bibitemdeclare{book}{Milner80}
\bibitem{Milner80}
\bibinfo{author}{Robin \surnamestart Milner\surnameend} (\bibinfo{year}{1980}):
  \emph{\bibinfo{title}{A Calculus of Communicating Systems}}.
\newblock {\sl \bibinfo{series}{LNCS}}~\bibinfo{volume}{92},
  \bibinfo{publisher}{Springer}, \doi{10.1007/3-540-10235-3}.

\bibitemdeclare{article}{mousavi2007sos}
\bibitem{mousavi2007sos}
\bibinfo{author}{MohammadReza \surnamestart Mousavi\surnameend},
  \bibinfo{author}{Michel~A. \surnamestart Reniers\surnameend} \&
  \bibinfo{author}{Jan~Friso \surnamestart Groote\surnameend}
  (\bibinfo{year}{2007}): \emph{\bibinfo{title}{SOS Formats and Meta-Theory: 20
  Years After}}.
\newblock {\sl \bibinfo{journal}{\abbrev{Theoretical Computer Science}{Theor.
  Comp. Sci.}}} \bibinfo{volume}{373}(\bibinfo{number}{3}), pp.
  \bibinfo{pages}{238--272}, \doi{10.1016/j.tcs.2006.12.019}.

\bibitemdeclare{phdthesis}{Peressotti}
\bibitem{Peressotti}
\bibinfo{author}{Marco \surnamestart Peressotti\surnameend}
  (\bibinfo{year}{2017}): \emph{\bibinfo{title}{Coalgebraic Semantics of
  Self-Referential Behaviours}}.
\newblock Ph.D. thesis, \bibinfo{school}{University of Udine},
  \doi{10.13140/rg.2.2.26899.07203}.

\bibitemdeclare{techreport}{PlotkinSOS}
\bibitem{PlotkinSOS}
\bibinfo{author}{Gordon~D. \surnamestart Plotkin\surnameend}
  (\bibinfo{year}{1981}): \emph{\bibinfo{title}{A Structural Approach to
  Operational Semantics}}.
\newblock \bibinfo{type}{DAIMI Report} \bibinfo{number}{FN-19},
  \bibinfo{institution}{Computer Science Department, Aarhus University}.

\bibitemdeclare{book}{riehl}
\bibitem{riehl}
\bibinfo{author}{Emily \surnamestart Riehl\surnameend} (\bibinfo{year}{2014}):
  \emph{\bibinfo{title}{Categorical Homotopy Theory}}.
\newblock {\sl \bibinfo{series}{New Mathematical
  Monographs}}~\bibinfo{volume}{24}, \bibinfo{publisher}{Cambridge University
  Press}, \doi{10.1017/CBO9781107261457}.

\bibitemdeclare{book}{SangioRutten}
\bibitem{SangioRutten}
\bibinfo{editor}{Davide \surnamestart Sangiorgi\surnameend} \&
  \bibinfo{editor}{Jan \surnamestart Rutten\surnameend}, editors
  (\bibinfo{year}{2011}): \emph{\bibinfo{title}{Advanced Topics in Bisimulation
  and Coinduction}}.
\newblock {\sl \bibinfo{series}{Cambridge Tracts in Theoretical Computer
  Science}}~\bibinfo{volume}{52}, \bibinfo{publisher}{Cambridge University
  Press}, \doi{10.1017/CBO9780511792588}.

\bibitemdeclare{inproceedings}{DBLP:conf/lics/Staton08}
\bibitem{DBLP:conf/lics/Staton08}
\bibinfo{author}{Sam \surnamestart Staton\surnameend} (\bibinfo{year}{2008}):
  \emph{\bibinfo{title}{General Structural Operational Semantics through
  Categorical Logic}}.
\newblock In: {\sl \bibinfo{booktitle}{Proc.\ 23rd \abbrev{Symposium on Logic
  in Computer Science}{Logic in Comp. Sci.}}}, pp. \bibinfo{pages}{166--177},
  \doi{10.1109/LICS.2008.43}.

\bibitemdeclare{inproceedings}{plotkin:turi:bialgebraic}
\bibitem{plotkin:turi:bialgebraic}
\bibinfo{author}{Daniele \surnamestart Turi\surnameend} \&
  \bibinfo{author}{Gordon~D. \surnamestart Plotkin\surnameend}
  (\bibinfo{year}{1997}): \emph{\bibinfo{title}{Towards a Mathematical
  Operational Semantics}}.
\newblock In: {\sl \bibinfo{booktitle}{Proc.\ 12th \abbrev{Symposium on Logic
  in Computer Science}{Logic in Comp. Sci.}}}, pp. \bibinfo{pages}{280--291},
  \doi{10.1109/LICS.1997.614955}.

\bibitemdeclare{article}{Weber:famfun}
\bibitem{Weber:famfun}
\bibinfo{author}{Mark \surnamestart Weber\surnameend} (\bibinfo{year}{2007}):
  \emph{\bibinfo{title}{Familial 2-functors and parametric right adjoints}}.
\newblock {\sl \bibinfo{journal}{Theory and Applications of Categories}}
  \bibinfo{volume}{18}(\bibinfo{number}{22}), pp. \bibinfo{pages}{665--732}.

\end{thebibliography}

\end{document}